\documentclass[reprint, aps, prx, tightenlines, notitlepage, footinbib,longbibliography,floatfix,nobalancelastpage, superscriptaddress]{revtex4-2}

\usepackage{cmap}
\usepackage[utf8]{inputenc}
\usepackage{float}
\usepackage{graphicx}  
\usepackage{physics}
\usepackage{esint}
\usepackage{amssymb}  
\usepackage{mathtools}
\usepackage{amsmath}
\usepackage{amsthm}
\usepackage{mathrsfs}
\usepackage{bm}
\usepackage{CJKutf8}

\usepackage[dvipsnames]{color,xcolor,colortbl}
\usepackage{hyperref}
\hypersetup{%
    colorlinks=true,
    linkcolor=NavyBlue,
    citecolor=NavyBlue,
    urlcolor=Bittersweet,
    plainpages=false
    }

\newtheorem{theorem}{Theorem}

\newtheorem*{result*}{Result}
\newtheorem{proposition}{Proposition}
\newtheorem{lemma}{Lemma}

\newtheorem*{corollary*}{Corollary}

\newtheorem*{claim*}{Claim}

\begin{document}


\begin{CJK}{UTF8}{gbsn}
\title{
Correlated Noise Estimation with Quantum Sensor Networks
}

\author{Anthony J. Brady}\email{ajbrad4123@gmail.com}
\affiliation{Joint Center for Quantum Information and Computer Science, NIST/University of Maryland, College Park, MD, 20742, USA}
\affiliation{Joint Quantum Institute, NIST/University of Maryland, College Park, MD, 20742, USA}

\author{Yu-Xin Wang (王语馨)}

\affiliation{Joint Center for Quantum Information and Computer Science, NIST/University of Maryland, College Park, MD, 20742, USA}

\author{Victor V. Albert}

\affiliation{Joint Center for Quantum Information and Computer Science, NIST/University of Maryland, College Park, MD, 20742, USA}

\author{Alexey V. Gorshkov}

\affiliation{Joint Center for Quantum Information and Computer Science, NIST/University of Maryland, College Park, MD, 20742, USA}
\affiliation{Joint Quantum Institute, NIST/University of Maryland, College Park, MD, 20742, USA}

\author{Quntao Zhuang}
\affiliation{Ming Hsieh Department of Electrical and Computer Engineering,
University of Southern California, Los Angeles, CA 90089, USA}
\affiliation{Department of Physics and Astronomy, University of Southern California, Los Angeles, CA 90089, USA}

\begin{abstract}
    We address the metrological problem of estimating collective stochastic properties imprinted on a network of quantum sensors. Canonical examples include center-of-mass quadrature fluctuations in a system of bosonic modes and correlated dephasing in an ensemble of qubits (e.g., spins), bosons, or fermions. We develop a theoretical framework to determine the limits of correlated (weak) noise estimation with quantum sensor networks and reveal the requirements for entanglement advantage. Notably, an advantage emerges from the synergistic interplay between quantum correlations of the sensors and ``classical'' correlations of the noises. We determine optimal entangled probe states and identify a sensing protocol---reminiscent of a many-body echo---that achieves the fundamental limits of measurement sensitivity for a broad class of problems, unveiling a route towards entanglement-enhanced metrology of correlated many-body phenomena.
\end{abstract}

\date{\today}

\maketitle
\end{CJK}


\textit{Introduction.---}Quantum metrology represents one of the most promising areas of quantum information science~\cite{Giovannetti2004BeatSQL,Giovannetti2006QuMetrology}, dedicated to estimating parameters encoded in the quantum state of a physical system, such as a spin (qubit), photonic mode, or mechanical oscillator. Employing quantum systems as sensors~\cite{Degen2017RMP_QuSensing} and engineering special probe states---such as Greenberger–Horne–Zeilinger (GHZ) states, spin-squeezed states, and squeezed vacuum states---enables estimation with precision surpassing classical methods. 

When estimating the parameter $\vartheta$ encoded in a quantum system through a unitary process, $\hat{U}_{\vartheta}=e^{-i\vartheta\hat{h}}$, the Heisenberg limit, ${{\rm Var}(\vartheta)\propto 1/\nu^2}$, establishes the ultimate precision bound, attainable via entangled quantum probes given $\nu$ copies of the unitary, $\hat{U}_{\vartheta}^{\otimes\nu}$~\cite{Giovannetti2006QuMetrology}. Here ${\rm Var}(\vartheta)$ denotes the mean squared error. This contrasts with typical shot-noise scaling, ${\rm Var}(\vartheta)\propto 1/\nu$, attainable via independent experiments and separable (non-entangled) quantum probes~\cite{Giovannetti2004BeatSQL,Giovannetti2006QuMetrology}. Prominent examples include: estimating the phase of a two-level system (e.g., a spin or qubit) for precision timing~\cite{Buifmmode199PRL_Clocks,Ludlow2015RMP_Clocks,Pezze2018RMP_atomic,Robinson2024NatPhys_Clocks} or the phase of a bosonic mode for interferometry~\cite{Holland1993PRL_phase,Tomohisa2007Sci_OpticalPhase,Wang2019NatComm_PhaseCQED,Deng2024NatPhys_largeFock}; estimating the displacement of an oscillator~\cite{Wolf2019MotionalFock,Mccormick2019Nat_ionOscillator} for gravitational wave astronomy~\cite{Caves1981LIGOsqueeze,Schnabel2010NatComm_QuMetrologyGW,LIGO2024Sci_squeeze}; and Hamiltonian learning~\cite{Huang2023PRL_LearnHamiltonian,Li2024npj_HLHamLearn,Ma2024LearnHamCompressed,Mirani2024LearnFermions}.

Quantum sensor networks (QSNs) represent a distinct paradigm of quantum metrology in which we aim to estimate, e.g., a collective property, such as a linear combination of local parameters. Entanglement between $K$ quantum sensors facilitates an enhancement (Heisenberg scaling with $K$) for these aggregrate queries~\cite{Eldredge2018PRA_SecureQSNs,Proctor2017QSNarxiv,Proctor2018PRL_qsn,Zhuang2018PRA_dqs,Ge2018PRL_qsn,Qian2019AnalyticFns,Qian2021FieldProps, Zhang2021DQSrvw}. Examples include distributed versions of standard estimation problems, featuring a network of quantum clocks~\cite{Komar2014QuClocks,Kessler2014PRL_QSNClocks,Nichol2022Nat_QSNelementaryClocks,Malia2022Nat_QSNclocks}, networked optical interferometry~\cite{Ge2018PRL_qsn,Guo2020Nat_DQSphase,Liu2021NatPh_DQSphase}, and collective displacement sensing~\cite{Zhuang2018PRA_dqs,Xia202PRL_RadioQSN,Xia2023OmechDQS,Gilmore2021IonEFieldQSN}.

Apart from  unitary parameter estimation, estimating stochastic properties (noise) is equally important. Canonical examples include: spin dephasing~\cite{Kolodynski2013Efficient_Noise,Pirandola2017PRL_AdaptiveLimits}, bosonic dephasing~\cite{Huang2024BoseDephasing}, and random displacement sensing~\cite{Ng2016PRA_NoiseSpectro,Gorecki2022SpreadChannel,Shi2023DMLimits,Tsang2023NoiseSpectr,Gardner2024StochEst} to name a few~\cite{Gorecki2024EnergyTime} (see also the related topic of multi-qubit noise characterization~\cite{Flammia2020PauliLearn,Harper2020PauliLearn, Chen2022:PauliChEst}). Noise estimation fundamentally differs from unitary parameter estimation, presenting its own distinct challenges and techniques~\cite{Pirandola2017PRL_AdaptiveLimits} (see also~\cite{Ji2008Programmable, Escher2011Nat_GenFramework,Dobrza2012ElusiveHeisenberg,Kolodynski2013Efficient_Noise, Demkowicz2014PRL_EntangleVsNoise, Jeske2014NJP_CorrMarkov, Sekatski2017Quantum_FullFast,Dobrza2017PRX_MarkovNoise,Zhou2018QECmetrology}). Given $\nu$ identical and independent copies of a noise channel, $\Phi_{\vartheta}^{\otimes\nu}$, the estimation precision of $\vartheta$ is constrained to shot-noise scaling with $\nu$~\cite{Ji2008Programmable,Dobrza2012ElusiveHeisenberg}. However, insights about independent noise channels do not directly translate to scenarios involving correlated noise.

Under what circumstances can we expect an entanglement advantage when estimating (spatially) correlated stochastic properties with QSNs, and how might we achieve this? Positive answers in this direction have implications for various fields---such as searching for new physics with quantum sensors~\cite{Kimball2023FundPhysAndSpin,YeZoller2024PRL_Essay,Bass2024NatRvw,Brady2022QuNetworkDMSearch,Chen2024PRL_DMQuEnhance,Ito2024JHEP_DMsearch}, collective force~\cite{Brady2023OmechArray,Xia2023OmechDQS} or electric and magnetic field sensing~\cite{Gilmore2021IonEFieldQSN,Ji2024NatPh_randEfield,Rovny2025}, physical-layer supervised learning with QSNs~\cite{Zhuang2019SLAEN_theory,Xia2021PRX_SLAEN}, and interrogating many-body quantum systems~\cite{Casola2018nvCM,Rovny2024NVforManyBody, Rovny2022NVcovariance, Cheng2025:mplex}. {In these instances, correlated noise arises from two-point correlations of the underlying ``reservoir" being probed (Fig.~\ref{fig:qsn-illustration}), whether that be a condensed-matter system or semi-classical field acting on the sensors.}

In this Letter, we develop a theoretical framework for estimating correlated stochastic processes with QSNs. Our analysis reveals that, given a $K$-sensor noise channel, classical spatial correlations between the noise processes at different sensor sites prove crucial for an entanglement advantage in estimation. Physically, the QSN may consist of qubits, fermions, or bosons. Nonetheless, our major findings are agnostic to the composition of the systems and, thus, apply broadly. To illustrate this versatility, we present concrete examples where an entanglement advantage in estimation emerges: including correlated spin, bosonic, and fermionic dephasing and random bosonic displacements. Our rigorous results focus on single-parameter estimation in the weak-noise regime [Eq.~\eqref{eq:channel_expansion}], while extensions to multi-parameter settings are discussed illustratively rather than proven generally.

\begin{figure}
    \centering
    \includegraphics[width=.95\linewidth]{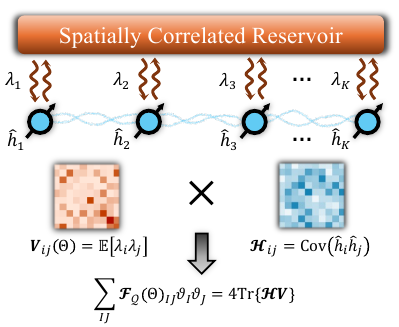}
    \caption{{Entangled sensors probe an extended reservoir. Quantum correlations of the quantum sensor network (viz., $\bm{\mathcal{H}}$) and ``classical" correlations of the reservoir (viz., $\bm V$) collude to enable an entanglement advantage (quantified by the QFI matrix $\bm{\mathcal{F}}_{\mathcal{Q}}$) in estimating aggregate parameters $\Theta=\{\vartheta_I\}$ of the reservoir.}}
    \label{fig:qsn-illustration}
\end{figure}

Problems in estimating correlated spin dephasing~\cite{Hainzer2024PRX_CorrSpec,Matsuzaki2018CollDephase,Wang2024ExpNoise,Dey2024EntangledFreqResol} and random displacements~\cite{Brady2022QuNetworkDMSearch,Brady2023OmechArray, Shi2023DMLimits} have been partially explored. However, estimating collective bosonic dephasing, to the best of our knowledge, has not been addressed, particularly in the energy-constrained setting. Recently, Huang \textit{et al}~\cite{Huang2024BoseDephasing} studied estimation of single-mode bosonic dephasing at infinite energy. Our results on correlated dephasing at finite energy, together with Ref.~\cite{Huang2024BoseDephasing}, sketch a more complete picture of the problem.

\textit{Theoretical framework.---}
Consider $K$ quantum sensors comprising a QSN, and associate to each sensor a local (Hermitian) generator $\hat{h}_j\,(j=1,\dots,K)$. For concreteness, we assume only one generator per sensor, so that all generators commute. The formalism readily extends to multiple non-Hermitian generators; see Appendix~\ref{app:channel_derive} for further discussion. Suppose the local generator $\hat{h}_j$ induces the translation $\lambda_j\in\mathbb{R}$ on the $j$th sensor, and let the the translations $\{\lambda_j\}_{j=1}^K$ be multivariate (e.g., Gaussian) random variables. We prioritize noise estimation and henceforth take $\mathbb{E}[\lambda_j]=0\,\forall\,j$. The translations correspond to (weak) fluctuations, which we package into the $K\times K$ covariance matrix, $\bm V>0$, with elements ${\bm V_{ij}=\mathbb{E}[\lambda_i\lambda_j]}$. [We disregard higher-order moments, presuming they are parametrically weaker.] Generally, there exist correlations between the translations at different sensors due to classically shared randomness.

Let $\rho\in \mathscr{H}^{\otimes K}$ denote the QSN probe state, with $\mathscr{H}$ the Hilbert space of a single sensor. The quantum channel $\Phi_{\bm V}$ encodes the fluctuations $\bm V$ onto the probe via $\Phi_{\bm V}(\rho)$. In this work, we crucially assume that the encoded state admits the approximate form 
\begin{equation}
\label{eq:channel_expansion}
   \Phi_{\bm V}(\rho)\approx \rho+ \sum_{i,j=1}^K\bm V_{ij}\left(\hat{h}_i\rho\hat{h}_j-\frac{1}{2}\acomm{\hat{h}_i\hat{h}_j}{\rho}\right).
\end{equation} 
Intuitively, this can be understood as representing many-body open-system dynamics within the Markov regime, where $\hat{h}_j$ are local jump operators and $\bm V_{ij}=\bm{\gamma}_{ij}\Delta t$, with $\bm{\gamma}_{ij}$ denoting many-body decoherence rates and $\Delta t$ an infinitesimal time interval. {The elements $\bm V_{ij}$ are then proportional to the two-point correlators of a quantum system that the sensors are probing~\cite{Casola2018nvCM, Rovny2024NVforManyBody} or, likewise, to the correlators of a (stochastic) semi-classical field that couples locally to each sensor generator $\hat{h}_i$. In this picture,} $\rho=\rho(t)$ and $\Phi_{\bm V}(\rho)=\rho(t+\Delta t)$. Alternatively, the expansion may describe perturbative random unitary evolution, where the unitary $\bigotimes_{j=1}^K e^{-i\lambda_j\hat{h}_j}$ acts according to probability $p(\lambda_1,\dots,\lambda_K)$, with zero mean and covariance $\bm V$. In this interpretation, the expansion is with respect to a small parameter $\varepsilon\ll 1$ such that $\mathbb{E}[\lambda_{i_1}\dots\lambda_{i_n}\hat{h}_{i_1}\dots\hat{h}_{i_n}]\sim\varepsilon^n$ $(n\geq 2)$. See the Appendix~\ref{app:channel_derive} for further motivation.

We aim to estimate a collection of $n$ aggregate parameters, $\Theta\coloneqq\{\vartheta_J\}_{J=1}^n$. We concentrate on two different cases: (i) The parameters are directly embedded into the covariance matrix, $\bm V(\Theta)$, thereby controlling elements of $\bm V$ (cf.~\cite{Eldredge2018PRA_SecureQSNs,Qian2021FieldProps}). (ii) The parameters are constructed from non-trivial combinations of the $\bm V_{ij}$'s [e.g., $\vartheta^2_{w}(\bm V)=\vec{w}^\top\bm V\vec{w}$], without explicit knowledge of $\bm V$ (cf.~\cite{Proctor2017QSNarxiv,Qian2019AnalyticFns, Bringewatt2021PRR}). For clarity, we refer to these settings as case (i) and case (ii).

To address these problems, we utilize quantum estimation theory~\cite{Paris2009QFI,Sidhu2020QFIrvw,Liu2020QFIM}. Given the encoded quantum data, $\Phi_{\bm V}(\rho)$, we perform measurements and construct estimates, $\{\check{\vartheta}_J\}_{J=1}^n$, from the measurement statistics. We quantify estimation performance by the mean-squared error, ${\rm Var}(\vartheta_J)\coloneqq \mathbb{E}[(\check{\vartheta}_J-\vartheta_J)^2]$. For unbiased estimation ($\mathbb{E}[\check{\vartheta}_J]=\vartheta_J$) the quantum Fisher information (QFI) bounds the mean-squared error from below~\cite{Paris2009QFI,Sidhu2020QFIrvw,Liu2020QFIM},
\begin{equation}\label{eq:var_qfi}
    {\rm Var}(\vartheta_I)\geq \nu^{-1}(\bm{\mathcal{F}}^{-1}_{\mathcal{Q}}(\Theta))_{II}\geq  \nu^{-1} (\bm{\mathcal{F}}_{\mathcal{Q}}(\Theta))_{II}^{-1},
\end{equation}
where $\bm{\mathcal{F}}_{\mathcal{Q}}(\Theta)$ is the QFI matrix for parameters $\Theta$ and $\nu$ denotes the number of independent repetitions. This bound does not guarantee achievability nor provide optimal measurement strategies, subjects that we examine further later.

To compute the QFI matrix, we employ the geometric relation between $\bm{\mathcal{F}}_{\mathcal{Q}}$ and the fidelity of quantum states through the Bures distance~\cite{Braunstein1994PRL_Bures}. 
Assuming $\Theta$ parametrize small perturbations to the identity channel, the Bures distance reads
\begin{equation}\label{eq:bures}
    \sum_{I,J=1}^n\left(\bm{\mathcal{F}}_{\mathcal{Q}}(\Theta)\right)_{IJ}\vartheta_I\vartheta_J=8\left(1-\sqrt{F(\rho,\rho_{\vartheta})}\right),
\end{equation}
where $\rho_{\vartheta}=\Phi_{\bm V}(\rho)$. 
Here, $F(\tau,\rho)={\rm Tr}[\sqrt{\sqrt{\tau}\rho\sqrt{\tau}}\,]^2$ is the fidelity.
We then use the channel approximation~\eqref{eq:channel_expansion} to expand the fidelity to leading order as $F(\rho,\rho_{\vartheta})\approx 1-\Tr{\bm V\bm{\mathcal{H}}}$, assuming $\Tr{\bm V\bm{\mathcal{H}}}\ll 1$ and a pure probe state which maximizes the QFI through a convexity argument; see Appendix~\ref{app:qfim_derive} for details.
This yields our main result, which applies to cases (i) and (ii) mentioned previously.

\begin{result*}\label{result:1}
    Consider the pure QSN probe $\rho=\dyad{\psi}$, and let $\bm{\mathcal{H}}$ denote the generator matrix of $\rho$, with elements $\bm{\mathcal{H}}_{ij}= \expval*{\hat{h}_i\hat{h}_j}-\expval*{\hat{h}_i}\expval*{\hat{h}_j}$. We deduce a direct relation between the QFI matrix $\bm{\mathcal{F}}_{\mathcal{Q}}(\Theta)$, the covariance matrix of the channel $\bm V$, and the generator matrix $\bm{\mathcal{H}}$,
    \begin{equation}\label{eq:qfim_Vh}
        \sum_{I,J=1}^n\left(\bm{\mathcal{F}}_{\mathcal{Q}}(\Theta)\right)_{IJ}\vartheta_I\vartheta_J \overset{!}{=}4\Tr{\bm V\bm{\mathcal{H}}}.
    \end{equation} 
\end{result*}

The degree of classical, spatial correlations between the noise sources thus dictates whether an entanglement advantage manifests.
On the one hand, when the noises are uncorrelated ($\bm V_{ij}=V_i\delta_{ij}$), the QFI matrix depends solely on the local generator variances, $\bm{\mathcal{H}}_{ii}={\rm Var}(\hat{h}_i)$, and entangled quantum probes bear no fruit. This aligns with established results in uncorrelated noise estimation~\cite{Hotta2006NbodyNoise, Kolodynski2013Efficient_Noise}, {which demonstrate that separable probes suffice for optimal estimation}. On the other hand, when spatial correlations between the noises exist, entanglement across the QSN affords enhanced precision for all types of random quantum processes.

\textit{Examples.---}
We first target an appealing class of problems that fall under case (i). Suppose $\bm V_{ij}=g^2\bm v_{ij}$, where $\bm v_{ij}$ are known. Physically, the numbers $\bm v_{ij}$ may be functions of device characteristics or geometrical functions that depend on the relative distances between the sensors, the distances from the sensors to the source of $g$ etc. This is relevant in practical scenarios, where the sensor array has been characterized and calibrated to, e.g., detect global random forces~\cite{Xia2023OmechDQS}, electric fields~\cite{Gilmore2021IonEFieldQSN,Ji2024NatPh_randEfield}, or to search for new physics~\cite{Brady2022QuNetworkDMSearch,Brady2023OmechArray,Kimball2023FundPhysAndSpin,YeZoller2024PRL_Essay,Bass2024NatRvw,Chen2024PRL_DMQuEnhance,Ito2024JHEP_DMsearch}. 

By consequence of Eq.~\eqref{eq:qfim_Vh}, we obtain the following.
\begin{corollary*}\label{cor:fisher_g}
    Suppose $\bm V_{ij}=g^2\bm v_{ij}$ with $g$ an unknown parameter and $\bm v_{ij}$ known. The QFI for estimating $g$ is
\begin{equation}\label{eq:qfi_g}
    \mathcal{F}_{\mathcal{Q}}(g)=4\Tr{\bm{\bm v\mathcal{H}}}.
\end{equation}
\end{corollary*}
\noindent
When $\bm v$ is rank 1 and permutation invariant (indicating maximal correlations), such that $\bm v=K\vec{u}\vec{u}^\top$ with $\vec{u}=(1,1,\dots,1)^\top/\sqrt{K}$, we anticipate the largest entanglement advantage. To formalize this, define the average generator, $\hat{H}_{{\rm avg}}\coloneqq \sum_{i=1}^K\hat{h}_i/K$. Per Eq.~\eqref{eq:qfi_g},
\begin{equation}\label{eq:qfi-Havg}
    \mathcal{F}_{\mathcal{Q}}(g)=4K^2{\rm Var}\big(\hat{H}_{{\rm avg}}\big).
\end{equation}
{To achieve Heisenberg scaling, $\mathcal{F}_{\mathcal{Q}} \propto K^2$, the quantity ${\rm Var}(\hat{H}_{\rm avg})$ must be a constant, i.e., independent of $K$. Since ${{\rm Var}(\hat{H}_{\rm avg}) = \sum_{i,j} {\rm Cov}(\hat{h}_i, \hat{h}_j)/K^2}$, this implies that each two-body correlator must be commensurate, which is possible via long-range many-body quantum correlations. For direct comparison, separable probes satisfy ${\rm Var}(\hat{H}_{\rm avg})\big|_{\rm sep} = \sum_{i}{\rm Var}(\hat{h}_i)/K^2$, which implies shot-noise scaling, $\mathcal{F}_{\mathcal{Q}}^{\rm sep}(g) \propto K$, at best.} We illustrate this for prototypical noise processes.

\subparagraph{Example 1.}
Consider qubit (or spin) dephasing with local generators $\hat{h}_i=\hat{Z}_i$, where $\hat{Z}_i$ is the Pauli-Z operator. Define the average spin ${\hat{Z}_{\rm avg}=\sum_{i=1}^K\hat{Z}_i/K}$. For entangled and separable strategies, respectively, we find 
\begin{align}
        \mathcal{F}_{\mathcal{Q}}^{\rm ent}(g)&=4K^2{\rm Var}\big(\hat{Z}_{\rm avg}\big) \leq 4K^2,\label{eq:spin_ent}
        \\ 
        \mathcal{F}_{\mathcal{Q}}^{\,\rm sep}(g)&=4K\left(\sum_{i=1}^K {\rm Var}(\hat{Z}_i)/K\right) \leq 4K.\label{eq:spin_sep}
\end{align}  
The optimal separable probe is the product state ${((\ket{\uparrow}+\ket{\downarrow})/\sqrt{2})^{\otimes K}}$. The optimal entangled probe is the GHZ state, $(\ket{\uparrow\uparrow\dots}+\ket{\downarrow\downarrow\dots})/\sqrt{2}$ (cf.~\cite{Matsuzaki2018CollDephase,Wang2024ExpNoise,Dey2024EntangledFreqResol}).

\subparagraph{Example 2.}
Consider bosonic dephasing with local generators $\hat{h}_i=\hat{n}_i$, where $\hat{n}_i$ is the occupation operator. Define the average occupation $\hat{n}_{\rm avg}=\sum_{i=1}^K\hat{n}_i/K$, and assume the following constraint on occupation (e.g., energy) fluctuations, $\expval*{\hat{n}_i\hat{n}_j}-\expval*{\hat{n}_i}\expval{\hat{n}_j}\leq \Bar{n}^2$ where $\Bar{n}=\expval*{\hat{n}_{\rm avg}}$. For entangled and separable strategies, respectively, we find
    \begin{align}
        \mathcal{F}_{\mathcal{Q}}^{\rm ent}(g)&= 4K^2 {\rm Var}(\hat{n}_{\rm avg})\leq 4K^2\Bar{n}^2, \label{eq:boson_ent} \\
        \mathcal{F}_{\mathcal{Q}}^{\,\rm sep}(g)&=4K\left(\sum_{i=1}^K {\rm Var}(\hat{n}_i)/K\right)\leq 4K \Bar{n}^2.\label{eq:boson_sep}           
    \end{align}  
Up to a small correction $\order{\Bar{n}-\lfloor\Bar{n}\rfloor}$, the optimal separable probe is the product state $\big((\ket{0}+\ket{N})/\sqrt{2}\big)^{\otimes K}$, where $N=\lfloor\Bar{n}\rfloor$. The optimal entangled probe is the bosonic GHZ state $(\ket{00\dots 0}+\ket{NN\dots N})/\sqrt{2}$, which is an entangled non-Gaussian state. In fact, to reach the limit implied by Eq.~\eqref{eq:boson_ent}, entangled non-Gaussian states are necessary; see Appendix~\ref{app:nogo_gaussian}. 

We readily incorporate fermionic dephasing~\cite{Gonzalez2023FermProcessor,Schuckert2024FermionFTQC}, which shares similarities to both spin dephasing [Eqs.~\eqref{eq:spin_ent} and~\eqref{eq:spin_sep}] and bosonic dephasing [Eqs.~\eqref{eq:boson_ent} and~\eqref{eq:boson_sep}]. Analogous to the bosonic case, $\hat{n}_i$ denotes the fermionic occupation operator responsible for dephasing on the $i$th fermionic mode. However, due to the Pauli exclusion principle, the fermionic occupation per mode is restricted to $N=1$. For an even number of modes, the optimal probe is the fermionic GHZ state.

\subparagraph{Example 3.}
Consider random bosonic displacements with local generators $\hat{h}_i=\hat{p}_i$, where $\hat{p}_i$ is the momentum operator. Define the average momentum $\hat{p}_{\rm avg}=\sum_{i=1}^K\hat{p}_i/K$, and assume a total occupation constraint, $\sum_{i=1}^K\expval{\hat{n}_i}\leq K\Bar{n}$, for any input probe. For entangled and separable strategies, respectively, we find
    \begin{align}
        \mathcal{F}_{\mathcal{Q}}^{\rm ent}(g)&= 4K^2{\rm Var}(\hat{p}_{\rm avg})\leq 8K^2(\Bar{n} + 1/2),\\
        \mathcal{F}_{\mathcal{Q}}^{\,\rm sep}(g) &=4K\left(\sum_{i=1}^K {\rm Var}(\hat{p}_i)/K\right) \leq 8K(\Bar{n}+1/2).
    \end{align}  
The optimal separable probe is a product of $K$ squeezed vacuum states, each with $\Bar{n}$ quanta~\footnote{A Fock state is 3dB shy of the QFI. This stems from the fact that estimating random single-quadrature displacements differs slightly from estimating displacements on both quadratures, e.g., position and momentum fluctuations. For the latter, both Fock states and squeezed vacuum states are optimal in the lossless regime~\cite{Wolf2019MotionalFock,Gorecki2022SpreadChannel}.}. The optimal entangled probe is a distributed squeezed vacuum state (cf.~\cite{Zhuang2018PRA_dqs}) of $K\Bar{n}$ quanta. References~\cite{Brady2022QuNetworkDMSearch,Brady2023OmechArray} explored this estimation problem with QSNs consisting of microwave resonators and opto-mechanical sensors but did not discuss ultimate limits nor schemes to approach them (though see Supplementary Note 6 of Ref.~\cite{Shi2023DMLimits}).

{For maximally symmetric noise correlations [Eq.~\eqref{eq:qfi-Havg}], GHZ-type states achieve Heisenberg scaling. Though our formalism [viz., Eqs.~\eqref{eq:qfim_Vh},~\eqref{eq:qfi_g}] applies to arbitrary initial states and serves to identify QSN states that optimize estimation precision for a given noise-correlation structure.}

\begin{figure}
    \centering
    \includegraphics[width=\linewidth]{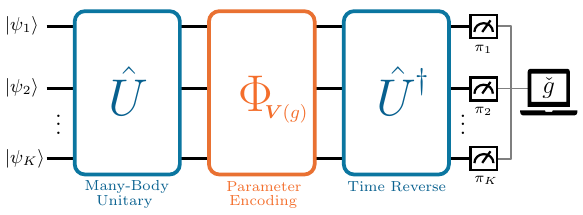}
    \caption{QSN Echo Protocol. Generate entangled probe, $\ket{\psi}$, by acting with many-body unitary, $\hat{U}$, on local states, $\{\ket{\psi_i}\}$. Quantum channel $\Phi_{\bm V}$ encodes parameter $g$ into probe. Revert the many-body unitary via $\hat{U}^\dagger$. Perform local projective measurements, $\pi_i=\dyad{\psi_i}$, and estimate $g$ from the measurement statistics.}
    \label{fig:qsn-echo}
\end{figure}

\textit{Sensing protocol.---}Equation~\eqref{eq:qfim_Vh} establishes fundamental precision limits of correlated noise estimation. However, a major challenge lies in designing sensing strategies that achieve these limits~\cite{Liu2020QFIM}.

Remarkably, for the class of single-parameter problems considered here, i.e., $\bm V=g^2\bm v$ and $\Tr{g^2\bm v\bm{\mathcal{H}}}\ll 1$, there exists an universally optimal sensing protocol that consists of projecting the output of the channel onto the input state (cf.~\cite{Gefen2019NatComm_Resolution, Gorecki2022SpreadChannel}). We present the protocol as a constructive achievability scheme rather than a new measurement primitive, demonstrating that the QFI bounds [Eq.~\eqref{eq:qfi_g}] can be saturated for single-parameter, correlated-noise sensing across various physical setups. We also describe multi-parameter extensions of the protocol to bosonic random-displacement sensing (see below) and multi-axis, collective Pauli noise estimation (see Appendix~\ref{app:spin-example}).

Consider the entangled probe $\rho=\dyad{\psi}$, and define the measurements $M_0=\dyad{\psi}$ and $M_1=I-M_0$. Using Eq.~\eqref{eq:channel_expansion}, we calculate the measurement probabilities: $p_0=\Tr{M_0\Phi_{\bm V}(\dyad{\psi})}\approx 1-g^2\Tr{\bm v\bm{\mathcal{H}}}$ and $p_1\approx g^2\Tr{\bm v\bm{\mathcal{H}}}$. We then compute the classical Fisher information for these measurements, $\mathcal{F}_{\mathcal{C}}(g)=\sum_{i=0}^1 (\partial_gp_i)^2/p_i\approx 4\Tr{\bm v\bm{\mathcal{H}}}$, which achieves the QFI in Eq.~\eqref{eq:qfi_g}. Notably, this procedure does not depend on problem specifics nor the physical systems involved.

The measurement protocol above resembles a metrological Loschmidt echo~\cite{Macri2016Echo,Yin2024echo,Gilmore2021IonEFieldQSN,Colombo2022QSNRamseyEcho}. Suppose we prepare the probe, $\psi$, by acting with the unitary circuit $\hat{U}$ on local quantum states, $\psi_i$, such that $\ket{\psi}=\hat{U}\big(\bigotimes_{i=1}^K\ket{\psi_i}\big)$. Then, we interpret the protocol as the following sequence: (1) Prepare the sensors in local product states, $\bigotimes_{i=1}^K\ket{\psi_i}$. (2) Entangle the sensors with the many-body operation $\hat{U}$. (3) Encode parameter $g$ via $\Phi_{\bm V}$. (4) Evolve the sensors under $\hat{U}^\dagger$, i.e., the time-reverse of $\hat{U}$. (5) Perform local projective measurements, $\pi_i=\dyad{\psi_i}$, and construct the binary POVM $\{M_0,I-M_0\}$, with $M_0=\bigotimes_i\pi_i$, to estimate $g$ from the measurement data. See Fig.~\ref{fig:qsn-echo}.

\textit{Multiple noise parameters.---}The case of multiple parameters~\cite{Ragy2016PRA_compatibility, Liu2020QFIM} poses significant challenges compared to single-parameter estimation. We observed entanglement advantage when estimating the single noise parameter $g$ [i.e., case (i)] and elaborated on achievability with specific sensing protocols. It is worthwhile to consider whether an advantage appears in multi-parameter problems. We argue that, at least in the case-study below an entanglement advantage is conceivable when restricting to $n< K$ aggregate parameters~\cite{Proctor2018PRL_qsn,Qian2019AnalyticFns,Bringewatt2021PRR}; see Appendix~\ref{app:spin-example} for a spin-based example using a purification of the Dicke (fully symmetric) subspace projector.

We construct the following multi-parameter example, which falls under case (ii). Consider the (user-specified) orthogonal matrix $\bm W=(\vec{w}_1,\dots,\vec{w}_K)^{\top}$, such that $\vec{w}_I^\top\vec{w}_J=\delta_{IJ}$. Let
$\bm V^{\prime}_{IJ}\coloneqq\vec{w}_I^\top\bm V\vec{w}_J$, and parametrize $\bm V^\prime$ via $\bm V^\prime_{IJ}=\xi_I\xi_J\mathscr{C}_{IJ}$, where $\xi_I^2\coloneqq \vec{w}_I^\top \bm V\vec{w}_I$. Here, $\mathscr{C}_{II}=1$ and $\mathscr{C}_{IJ}\leq 1$ otherwise~\footnote{There are also negativity conditions on $\mathscr{C}_{IJ}$ set by ${\bm V^\prime> 0}$, which is unimportant to our discussion.}, and we assume no knowledge of the elements of $\bm V$. We are interested in estimating aggregate parameters $\Xi\coloneqq\{\xi_I\}_{I=1}^K$, which describe fluctuations of the $I$th collective mode. The numbers $\mathscr{C}_{IJ}$ ($I\ne J$) represent correlation coefficients, which we treat as nuisance parameters here. We state the following claim (see Appendices for details).
\begin{claim*}
    The diagonal elements of the QFI matrix for the parameter set of non-local fluctuation $\Xi$ are
        \begin{equation}\label{eq:multiparam_bd}
        (\bm{\mathcal{F}}_{\mathcal{Q}}(\Xi))_{II}=4\vec{w}_I^\top\bm{\mathcal{H}}\vec{w}_I.
    \end{equation}
\end{claim*} 
Equation~\eqref{eq:multiparam_bd} suggests a simultaneous entanglement advantage in estimating all parameters in $\Xi$. However, the QFI bounds the error from below [see Eq.~\eqref{eq:var_qfi}], and the corresponding bound is generically not saturable for multiple parameters~\cite{Liu2020QFIM}. Though, if only the subset of parameters $\Xi_{n}\subset \Xi$ ($n< K$) interest us, it seems plausible to achieve a \textit{simultaneous} entanglement advantage ($\propto\! K/n$) over separable probes for all parameters in $\Xi_{n}$. We support this conjecture by an example involving correlated quadrature fluctuations in a system of $K$ oscillators; see Appendix~\ref{app:spin-example} for a spin-based example. 

Consider the random displacement channel $\Phi_{\bm V}$, where $\bm V$ denotes the covariance of displacements generated by local momenta $\{\hat{p}_i\}_{i=1}^K$. Let $\bm V^\prime_{IJ}=\vec{w}_I^\top\bm V\vec{w}_J$ represent non-local fluctuations generated by collective momenta $\{\hat{P}_I\}_{I=1}^K$, where $\hat{P}_I=\sum_i(\vec{w}_I)_i\hat{p}_i$. In other words, $\bm V^\prime_{II}\eqqcolon\xi_I^2$ symbolizes displacement fluctuations along the $I$th collective mode of the QSN. Equation~\eqref{eq:multiparam_bd} implies $\mathcal{F}^{\rm ent}(\xi_I)\propto{\rm Var}(\hat{P}_I)$, while for separable strategies, $\mathcal{F}^{\rm sep}(\xi_I)\propto\sum_{i=1}^{K}\bm W_{Ii}^2{\rm Var}(\hat{p}_i)$. We direct our attention to simultaneously estimating all parameters in the the subset $\Xi_{n}=\{\xi_I\}_{I=1}^{n}$ with $n<K$, for a fixed total occupation, $K\Bar{n}$. For the entangled strategy, we populate (e.g., squeeze) the $n$ collective modes---thus forming a continuous-variable entangled state of the QSN~\cite{Zhuang2018PRA_dqs}---with ${\rm Var}(\hat{P}_I)\approx K \Bar{n}/n$ for every $I\in\{1,\dots,n\}$~\cite{Zhuang2019SLAEN_theory}. Contrariwise, the separable strategy requires populating all $K$ modes with $\Bar{n}$ quanta each. It follows that $\mathcal{F}^{\rm ent}/\mathcal{F}^{\rm sep}\propto K/n$, suggesting a simultaneous entanglement advantage over separable strategies for \textit{every} parameter $\xi_I\in\Xi_{n}$. 

A variant of the echo sensing protocol realizes the above advantage~\footnote{The variation of the echo protocol comes in the final measurement step, which now involves $n$ independent POVMs, $\bigoplus_{I=1}^n\{M_{0}^{(I)},I-M_{0}^{(I)}\}$ used to simultaneously estimate each $\xi_I\in\Xi_{n}$. For (distributed) squeezed vacuum inputs, the measurements consist of anti-squeezing followed by photon counting~\cite{Gorecki2022SpreadChannel,Shi2023DMLimits,Tsang2023NoiseSpectr}.}. This resembles the original distributed sensing scheme of Refs.~\cite{Zhuang2018PRA_dqs,Zhuang2019SLAEN_theory} but, crucially, employs anti-squeezing and photon counting, rather than linear detection, on each mode at the measurement end.

{This example illustrates the structure supporting an entanglement advantage in multi-parameter correlated-noise estimation. However, a full derivation for $n<K$ parameters in generic settings remains outstanding. Further investigation is needed to characterize measurement compatibility for multiple parameters and pin down the regimes for entanglement advantage in various systems.}

\textit{Discussion.---}We have shown how quantum correlations within the QSN and spatial correlations within the noise process must collude to enable an entanglement advantage in noise estimation [viz., Eq.~\eqref{eq:qfim_Vh}]. The noise correlations are paramount, since entanglement alone does not improve the estimation of uncorrelated noises~\cite{Hotta2006NbodyNoise, Kolodynski2013Efficient_Noise}.

It is instructive to compare correlated noise estimation to conventional unitary estimation. In the unitary case, having $K$ copies of the unitary channel $\hat{U}^{\otimes K}$ and entanglement between the $K$ sensors improves estimation of the unitary process, yielding Heisenberg scaling with $K$~\cite{Giovannetti2006QuMetrology}. By contrast, probing $K$ copies of a single (i.e., uncorrelated) noise channel $\varphi^{\otimes K}$ yields shot-noise scaling with $K$ {when estimating the noise process}~\cite{Ji2008Programmable,Dobrza2012ElusiveHeisenberg}. {However, as we have shown, a correlated $K$-sensor noise channel $\Phi$, such that $\Phi \neq \varphi^{\otimes K}$, enables Heisenberg scaling with $K$ in noise estimation when probed with an entangled QSN state such as a GHZ state.}

Another distinction lies in decoherence effects. In unitary estimation, decoherence generated by the same Hamiltonian as the unitary process (i.e., parallel decoherence~\cite{Sekatski2017Quantum_FullFast}) inhibits Heisenberg scaling~\cite{Sekatski2017Quantum_FullFast,Dobrza2017PRX_MarkovNoise, Zhou2018QECmetrology}. Whereas, for correlated noise estimation, an entanglement advantage persists in the face of parallel decoherence. However, parallel decoherence does inevitably conjure Rayleigh's curse~\cite{Tsang2016Superresolution,Gardner2024StochEst}---that is, the QFI vanishes as the signal tends to zero; see Appendix~\ref{app:curse} for further discussion. 

{These findings highlight several avenues for future research. One direction is to extend the present framework to spatio-temporally correlated (non-Markovian) noise~\cite{Chin2012:nonMarkvoQuMetrology, White2022:nonMarkovQPT}, which requires a continuous-time treatment of the QSN and may benefit from temporal correlations in the probe state~\cite{Knaut2024:EarlyLateEntanglement}.} Developing a more comprehensive view of multi-parameter problems will be valuable. Assessing what role quantum control has to play in these contexts should prove fruitful~\cite{Sekatski2017Quantum_FullFast, Gardner2024StochEst, Shi2024MitRayleigh, Bringewatt2024, Gardner2025LindbladEst}. Evaluating the impact of decoherence, with distinct generators from the signal, offers meaningful insights~\cite{Gardner2025LindbladEst}. From a physics perspective, applications to searches for new physics~\cite{Brady2022QuNetworkDMSearch,Brady2023OmechArray,Kimball2023FundPhysAndSpin,YeZoller2024PRL_Essay,Bass2024NatRvw,Chen2024PRL_DMQuEnhance,Ito2024JHEP_DMsearch} and many-body quantum metrology~\cite{Casola2018nvCM, Rovny2024NVforManyBody, Montenegro2024ManyBodyMetrology, Rovny2025} exemplify compelling avenues for exploration and practical realization. 

\textit{Acknowledgments.---}The authors acknowledge Jacob Bringewatt for helpful conversations.
A.J.B.~acknowledges support from the NRC Research Associateship Program at NIST; A.J.B.~was supported by ONR N00014-23-1-2296 at USC for the initial part of the work. 
Y.-X.W.~acknowledges support from a QuICS Hartree Postdoctoral Fellowship. 
Q.Z.~acknowledges NSF (CCF-2240641, OMA-2326746, 2350153), ONR N00014-23-1-2296, ONR MURI N000142612102, AFOSR MURI FA9550-24-1-0349 and DARPA (HR00112490362, HR00112490453, D24AC00153-02). A.V.G.~was supported in part by ONR MURI, AFOSR MURI, DARPA SAVaNT ADVENT, ARL (W911NF-24-2-0107), NQVL:QSTD:Pilot:FTL, NSF QLCI (award No.~OMA-2120757), NSF STAQ program, and DoE ASCR Quantum Testbed Pathfinder program (awards No.~DE-SC0019040 and No.~DE-SC0024220). A.V.G.~also acknowledges support from the U.S.~Department of Energy, Office of Science, National Quantum Information Science Research Centers, Quantum Systems Accelerator (QSA) and from the U.S.~Department of Energy, Office of Science, Accelerated Research in Quantum Computing, Fundamental Algorithmic Research toward Quantum Utility (FAR-Qu).  \\[-0.5em]


\bibliography{main}


\clearpage
\appendix
\onecolumngrid

\section*{Appendices}
Here we provide details otherwise omitted in the main text. In Section~\ref{app:channel_derive}, we motivate the form of the approximate quantum channel utilized throughout the article. In Section~\ref{app:qfim_derive}, we derive the Result [Eq.\ (\ref{eq:qfim_Vh})] of the main text, i.e.~the direct relation between the QFI matrix and the noise covariance $\bm V$ of the quantum channel $\Phi_{\bm V}$. In Section~\ref{app:nogo_gaussian}, we derive a no-go result for entanglement enhanced dephasing estimation with entangled states generated by (passive) linear optics. In Section~\ref{app:multi-parameter}, we derive the QFI matrix for multiple noise parameters (see Claim [Eq.\ (\ref{eq:multiparam_bd})] of main text). In Section~\ref{app:curse}, we elaborate on how parallel decoherence triggers Rayleigh's curse in noise estimation problems. Finally, in Section~\ref{app:spin-example}, we analyze an explicit multi-parameter collective dephasing example using a spin-based QSN, highlighting the possibility of simultaneous estimation at the Heisenberg limit.

\section{Modeling the Approximate Noise Channel \label{app:channel_derive}}

In this section, we motivate the form of the approximate quantum channel [Eq.~\eqref{eq:channel_expansion}] analyzed throughout the main text. Recall the approximate form of the QSN noise channel $\Phi_{\bm V}$, which we write out here for convenience,
\begin{equation}
\label{appeq:ch_expand}
   \Phi_{\bm V}(\rho)\approx \rho+ \sum_{i,j=1}^K\bm V_{ij}\left(\hat{h}_i\rho\hat{h}_j-\frac{1}{2}\acomm{\hat{h}_i\hat{h}_j}{\rho}\right),
\end{equation} 
with $\hat{h}_j$ Hermitian. This formulation extends to non-Hermitian generators since $\hat{A}=\hat{h}+i\hat{f}$ for any operator $\hat{A}$, with $\hat{h}$ and $\hat{f}$ Hermitian; see the paragraph at the end of this section for further discussion.

We consider two physical, yet fairly generic, settings where this approximation applies: (a) short time open-system dynamics and (b) weak, random unitary channels. The explicit examples of spin dephasing, bosonic (fermionic) dephasing, and bosonic random displacements considered in the article fall within either setting.

\paragraph{Open-System Dynamics.}
The expansion~\eqref{appeq:ch_expand} mimics open-system dynamics described through the Lindblad master equation~\cite{Manzano2020LindbladIntro}. To reveal the correspondence, we first write the (Markovian) master equation in its general form,
\begin{equation}
\label{appeq:master_eq}
    \partial_t\rho=-i\comm*{\hat{H}_S}{\rho} + \sum_{i,j}\bm{\gamma}_{ij}\left(\hat{A}_j\rho\hat{A}_i^\dagger-\frac{1}{2}\acomm{\hat{A}_i^\dagger\hat{A}_j}{\rho}\right),
\end{equation}
where $\hat{H}_S$ is the system Hamiltonian, $\bm{\gamma}$ is the positive semi-definite matrix symbolizing the system's decoherence processes (per unit time), and $\{\hat{A}_j\}$ are jump operators. For purely noisy dynamics, as considered here, we {assume non-interacting sensors in the QSN} and thus take $\hat{H}_S=0$. Further, suppose the jump operators are Hermitian, i.e. $\hat{A}_i^\dagger=\hat{A}_i$ (though this may be generalized). Consider an infinitesimal time interval $\Delta t$, such that ${\partial_t\rho\approx (\rho(t+\Delta t)-\rho(t))/\Delta t}$. We then make the following correspondences, $\Phi_{\bm V}(\rho)= \rho(t+\Delta t)$, $\hat{h}_i=\hat{A}_i$, and $\bm V_{ij}=\bm{\gamma}_{ij}\Delta t$. Hence, estimation problems associated with $\bm V$ equate to estimation problems associated with the decoherences $\bm{\gamma}$ over a short time $\Delta t$ that roughly satisfies $\Delta  t\ll 1/\big(\sum_{i,j}\bm{\gamma}_{ij}{\rm Tr}(\hat{h}_i\hat{h}_j\rho)\big)$.

{\subparagraph{Dyson series expansion.}
At a physical level, the decoherence rates $\bm\gamma$ that appear in a Lindblad description are nothing more than the two-point correlators of the ``reservoir" (or environment, we use these terms interchangeably) degrees of freedom that the sensors are probing. The reservoir could be a many-body quantum system itself or a semi-classical field. To further illustrate this point, we perturbatively analyze unitary coupling between each sensor coupled locally to the same (Markovian) reservoir, which ultimately leads to the same description above. 

We label the Hilbert spaces as $S=\bigotimes_{i=1}^K S_i$ for the sensors and $E=\bigotimes_{i=1}^K E_i$ for the environment. We assume that the initial state factorizes as $\rho_{S}\otimes\tau_E$, where $\rho_S$ denotes the initial QSN state at $t=0$ and $\tau_E$ the initial state of the environment. We model the joint evolution as a spatially local bilinear interaction Hamiltonian $\hat{H}_{SE}$, represented in a time-ordered unitary $\hat{U}_{SE}$, which we expand perturbatively in a Dyson series. This makes transparent how, after tracing out the environment, the sensor dynamics acquire Lindblad-type terms whose coefficients are set precisely by the environmental two-point correlators.

The joint state after interaction is 
\begin{equation}
    \rho_{S E}(\Theta)=\hat{U}_{SE}(\Delta t)(\rho_{S }\otimes\tau_E)\hat{U}_{SE}^\dagger(\Delta t).
\end{equation}
The (time ordered) interaction unitary is 
\begin{equation}
    \hat{U}_{SE}(\Delta t)=\mathcal{T}\exp[-i\int_0^{\Delta t}\dd{t}\hat{H}_{SE}(t)]
\end{equation}
with $\hat{H}_{SE}(t)$ the time-dependent interaction operator. We model the coupling as a spatially local bi-linear interaction, 
\begin{equation}
    \hat{H}_{SE}(t)= \sum_i \hat{H}_i(t) =\sum_i g_i\hat{h}_i\hat{e}_i(t),
\end{equation}
where the index $i$ labels the $i$th sensor and $i$th local environmental degree of freedom; $g_i\in\mathbb{R}$ characterizes the interaction strength; $\hat{h}_i$ is the time-independent (in the rotating frame) generator for the $i$th sensor; and $\hat{e}_i(t)$ is the time-dependent environment operator. We assume the environment is Markovian (i.e., no temporal correlations) such that $\expval*{\hat{e}_i(\tau)\hat{e}_i(0)}\propto \delta(\tau)$; though we do allow for spatial correlations between environmental subsystems.

We now expand $\hat{U}_{SE}(\Delta t)$ in a Dyson series,
\begin{equation}
    \hat{U}_{SE}(\Delta t)\approx \hat{I} - i\int_{0}^{\Delta t}\dd{t}\sum_j\hat{H}_j(t) -\frac{1}{2}\int_{0}^{\Delta t}\dd{t}\int_{0}^t\dd{t'}\sum_{j,k}\hat{H}_j(t)\hat{H}_{k}(t'),
\end{equation}
with the sum taken over all sensors in the network. From here we obtain an expression for the joint state post-interaction:
\begin{multline}
    \rho_{SE}(\Theta)\approx \rho_S\otimes\tau_E - i\int_{0}^{\Delta t}\dd{t}\sum_j\comm{\hat{H}_j(t)}{\rho_{S}\otimes\tau_E} \\ + \int_{0}^{\Delta t}\dd{t}\int_{0}^t\dd{t'}\sum_{j,k}\left(\hat{H}_j(t)\rho_{S}\otimes\tau_E\hat{H}_{k}(t')-\frac{1}{2}\acomm{\hat{H}_j(t)\hat{H}_{k}(t')}{\rho_{S}\otimes\tau_E}\right).
\end{multline}
We only access the reservoir indirectly by measuring the sensors. Thus, we must trace over the reservoir degrees of freedom and find an expression solely for the quantum sensor network state, $\rho_{S}(\Theta)=\Tr_E\{\rho_{SE}(\Theta)\}$. We assume a zero-mean Markovian reservoir but allow for spatial correlations between subsystems, such that 
\begin{equation}
    \Tr\{\tau_E\hat{e}_i\}=0 \qq{and} \Tr\{\tau_E\hat{e}_i(\tau)\hat{e}_{j}(0)\}=\mathscr{E}_{ij}\delta(\tau),
\end{equation}
with $\mathscr{E}_{ij}$ representing the (spatial) correlation matrix of the reservoir. The following relations then hold:
\begin{align}
    \Tr_E\left(\comm{\hat{H}_j}{\rho_{S}\otimes\tau_E}\right)&=0,\label{eq:nodrive}\\
    \Tr_E\left(\hat{H}_j(t)\rho_{S}\otimes\tau_E\hat{H}_{k}(t')\right)&=g_jg_k\mathscr{E}_{jk}\hat{h}_j\rho_{S}\hat{h}_k\delta(t-t'),\\
    \Tr_E\left(\hat{H}_j(t)\hat{H}_{k}(t')\rho_{S}\otimes\tau_E\right)&=g_jg_k\mathscr{E}_{jk}\hat{h}_j\hat{h}_k\rho_{S}\delta(t-t').
\end{align}
The first equality implies that the environment does not coherently drive the sensors, while the second and third equalities follow from the Markovian assumption. Let $\bm{\gamma}_{jk}\coloneqq g_jg_k\bm{\mathscr{E}}_{jk}$. Given that $\hat{h}_i$ are time-independent operators, we deduce, after a trivial time integral, the approximate form of the output QSN state as
\begin{align}
    \rho_{S}(\Theta) \approx \rho_{S} &+  \Delta t\sum_{j,k}\bm{\gamma}_{jk}\left(\hat{h}_j\rho_{S}\hat{h}_k-\frac{1}{2}\acomm{\hat{h}_j\hat{h}_k}{\rho_{S}}\right),\label{eq:rho_oscillator}
\end{align}
which matches Eq.~\eqref{appeq:ch_expand}. In this regime, single-jump processes govern the QSN evolution, and the imprint of the interaction accumulates linearly with the interaction time $\Delta t$, consistent with Markovian dynamics. 
}

\paragraph{Weak, Random Unitary Channels.}
Consider the set of Hermitian generators $\hat{\vec{h}}=(\hat{h}_1, \hat{h}_2,\dots,\hat{h}_K)$ which induce shifts $\vec{\lambda}=(\lambda_1,\lambda_2,\dots,\lambda_K)$ on a quantum system described by the state $\rho$. For simplicity, we assume the generators are independent of $\vec{\lambda}$, while the $\lambda_j$'s are stochastic and described by the probability distribution, $p(\vec{\lambda})$. We associate a single realization of shifts with the parametrized unitary,
\begin{equation}\label{eq:unitary_h}
    \hat{U}_{\vec{\lambda}} = e^{-i\vec{\lambda}^\top\hat{\vec{h}}}.
\end{equation}
The following quantum channel describes the (random) evolution of the quantum system,
\begin{equation}\label{eq:probU_channel}
    \Phi(\rho)=\int\dd{\vec{\lambda}}p(\vec{\lambda})\hat{U}_{\vec{\lambda}}\rho\hat{U}_{\vec{\lambda}}^\dagger,
\end{equation}
Herein, we focus on dynamics associated with the mean $\vec{\mu}=\mathbb{E}[\vec{\lambda}]$ and covariance matrix $\bm V_{ij}=\mathbb{E}[\Delta\vec{\lambda}_i\Delta\vec{\lambda}_j]$, where $\Delta\vec{\lambda}\coloneqq\vec{\lambda}-\vec{\mu}$. The mean determines the unitary part of the evolution, while $\bm V$ describes fluctuations about the mean. We do not concern ourselves with higher order moments, supposing that fluctuations are weak.
\begin{lemma}\label{lemma:expand}
Assume fluctuations weakly perturb the quantum system, $\rho$, and the ``energy'' of the system is sufficiently small, such that $\mathbb{E}[\Delta\vec{\lambda}_{i_1}\dots\Delta\vec{\lambda}_{i_n}\hat{h}_{i_1}\dots\hat{h}_{i_n}]\sim\varepsilon^n$ for $n\geq 2$ and $\varepsilon\rightarrow 0$. If (1) the generators commute to a constant ($\comm*{\hat{h}_i}{\hat{h}_j}={\bm c}_{ij}\hat{I}$) or (2) the channel mean vanishes ($\vec{\mu}=0$), then,
\begin{equation}
\label{eq:simulable_expansion}
    \Phi_{\bm V}(\rho)\coloneqq \hat{U}_{\vec{\mu}}^\dagger\Phi(\rho)\hat{U}_{\vec{\mu}}\approx \rho+ \sum_{i,j}\bm V_{ij}\left(\hat{h}_i\rho\hat{h}_j-\frac{1}{2}\acomm{\hat{h}_i\hat{h}_j}{\rho}\right) + \order{\varepsilon^3}.
\end{equation}
\end{lemma}
\begin{proof}
We adopt Einstein summation convention, i.e.~repeated indices are summed over. Further, we assume that fluctuations are sufficiently small and expand to first non-trivial order around the mean, $\mathbb{E}[\vec{\lambda}]=\vec{\mu}$. For brevity, we define the fluctuation $\Delta\vec{\lambda}_j\coloneqq \vec{\lambda}_j-\vec{\mu}_j$, such that $\mathbb{E}[\Delta\vec{\lambda}]=0$ and $\bm V_{ij}=\mathbb{E}[\Delta\vec{\lambda}_i\Delta\vec{\lambda}_j]$.

Consider the unitary operator from Eq.~\eqref{eq:unitary_h} and rewrite it as follows:
\begin{equation}
    \hat{U}_{\vec{\lambda}}=e^{-i\vec{\lambda}^\top\hat{\vec{h}}}=e^{-i\Delta\vec{\lambda}^\top\hat{\vec{h}}-i\vec{\mu}^\top\hat{\vec{h}}},
\end{equation}
where $\Delta\vec{\lambda}=\vec{\lambda}-\vec{\mu}$. To expand in small fluctuations about the mean, use the Zassenhaus formula,
\begin{equation}
    e^{\hat{A}+\hat{B}}=e^{\hat{A}}e^{\hat{B}}e^{-\frac{1}{2}\comm{\hat{A}}{\hat{B}}}e^{-\frac{1}{6}\left(2\comm{\hat{B}}{\comm{\hat{A}}{\hat{B}}}+\comm{\hat{A}}{\comm{\hat{A}}{\hat{B}}}\right)}\times\dots,
\end{equation}
where ellipsis denote nested commutators of higher order. We attend to case (1) where the commutators between the generators equals a constant. The result for case (2), i.e.~zero mean $\vec{\mu}=0$, presents itself straightforwardly.

Let $\comm*{\hat{h}_i}{\hat{h}_j}=\bm{c}_{ij}$, where ${\bm c}_{ij}^*=-{\bm c}_{ij}$ are purely imaginary, and take $\hat{A}= -i\vec{\mu}^\top\hat{\vec{h}}$, and $\hat{B}= -i\Delta\vec{\lambda}^\top\hat{\vec{h}}$. Then only the first-order commutator matters and equals $\comm*{\hat{A}}{\hat{B}}=-\vec{\mu}^{\top}\bm c\Delta\vec{\lambda}$, which is purely imaginary. We deduce that
\begin{equation}
    \hat{U}_{\vec{\lambda}}\rho\hat{U}_{\vec{\lambda}}^{\dagger}=\hat{U}_{\vec{\mu}}\hat{U}_{\Delta\vec{\lambda}}\rho\hat{U}_{\Delta\vec{\lambda}}^\dagger\hat{U}_{\vec{\mu}}^\dagger.
\end{equation}
Expand the fluctuations as
\begin{equation}
    \hat{U}_{\Delta\vec{\lambda}}\approx\hat{I}-i\Delta\vec{\lambda}_j\hat{h}_j -\frac{1}{2}\Delta\vec{\lambda}_j\Delta\vec{\lambda}_k\hat{h}_j\hat{h}_k.
\end{equation}
The following set of equalities then hold to $\order{\varepsilon^3}$:
\begin{align}
    \Phi(\rho)
    &=\int\dd{\vec{\lambda}}p(\vec{\lambda})\hat{U}_{\vec{\lambda}}\rho\hat{U}_{\vec{\lambda}}^\dagger \\
    &\approx\int\dd{\vec{\lambda}}p(\vec{\lambda})\hat{U}_{\vec{\mu}}\Bigg(\rho-i\Delta\vec{\lambda}_j\hat{h}_j\rho + i\Delta\vec{\lambda}_j\rho\hat{h}_j + \Delta\vec{\lambda}_j\Delta\vec{\lambda}_k\hat{h}_j\rho\hat{h}_k \\  & 
    \hspace{13em} -\frac{1}{2}\Delta\vec{\lambda}_j\Delta\vec{\lambda}_k\hat{h}_j\hat{h}_k\rho - \frac{1}{2}\Delta\vec{\lambda}_j\Delta\vec{\lambda}_k\rho\hat{h}_j\hat{h}_k \Bigg)\hat{U}_{\vec{\mu}}^\dagger \nonumber\\
    &= \hat{U}_{\vec{\mu}}\left(\rho + {\bm V}_{jk}\hat{h}_j\rho\hat{h}_k - \frac{1}{2}{\bm V}_{jk}\rho\hat{h}_j\hat{h}_k - \frac{1}{2}{\bm V}_{jk}\hat{h}_j\hat{h}_k\rho\right)\hat{U}_{\vec{\mu}}^\dagger,
\end{align}
which concludes the proof.
\end{proof}
Further comments about the random unitary channel expansion are in order:
\begin{itemize}
    \item We properly understand the operator expansion in terms of measurement outcomes. Consider measurements $\{M_j\}$ with $\sum_jM_j=I$. Then, $\Tr\{M_j\Delta\lambda^m\hat{h}^m\rho\}\lesssim \varepsilon^m$. For brevity, we have let $\Delta\lambda^m$ and $\hat{h}^m$ denote any $m$th order product of fluctuations and generators, e.g., $\Delta\lambda^2\sim \bm V$ and  $\hat{h}^2\sim \hat{h}_i\hat{h}_j$, respectively. 
    \item The analysis above applies to programmable (or classically simulable) channels~\cite{Ji2008Programmable,Dobrza2012ElusiveHeisenberg} in the regime of weak noise and finite energy. Metrological bounds for programmable channels, as well as more general noise channels~\cite{Pirandola2017PRL_AdaptiveLimits}, typically apply in asymptotic regimes, e.g.~infinite-energy, and do not necessarily entail whether entangled probes are necessary to attain the precision limits derived therefrom.
    \item For simplicity, we have assumed one Hermitian generator per sensor. We briefly comment how our formalism can be extended to non-Hermitian generators: Consider a single non-Hermitian generator $\hat{A}_j$ per sensor associated with a complex translation $\alpha_j$, so that $ \hat{U}_{\vec{\alpha}} = e^{-i(\vec{\alpha}^\top\hat{\vec{A}}+\text{h.c.})}$. Any non-Hermitian operator can be written as $\hat{A}_j = \hat{h}_j + i \hat{f}_j$, with $\hat{h}_j,\hat{f}_j$ Hermitian, and similarly $\alpha_j = \lambda_j + i \nu_j$ with $\lambda_j,\nu_j$ real (possibly correlated) random variables. One can then map the single non-Hermitian generator with one complex translation to two Hermitian generators with two real translations. The covariance matrix $\bm V$ then becomes $2K\times 2K$, containing variances in $\lambda_j$ and $\nu_j$ and any cross-correlations. The Hamiltonian matrix similarly expands to include intra- and inter-sensor components for the pairs $(\hat{h}_j, \hat{f}_j)$. This extension preserves the structure of the formalism and the entanglement-scaling arguments, though introduces slightly more complexity due possible measurement incompatibility, even at the single-sensor level, due to single-site random variables $\lambda_j,\nu_j$~\cite{Ragy2016PRA_compatibility}.
\end{itemize}

\section{Derivation of the Main Result \label{app:qfim_derive}}

In this section, we prove the Result [Eq.~(\ref{eq:qfim_Vh})] of the main text, which we restate here for convenience: Consider the QSN probe $\rho=\dyad{\psi}$, and let $\bm{\mathcal{H}}$ denote the ``Hamiltonian'' generator matrix of $\rho$, with elements $\bm{\mathcal{H}}_{ij}= \expval*{\hat{h}_i\hat{h}_j}-\expval*{\hat{h}_i}\expval*{\hat{h}_j}$. Let $\Theta=\{\vartheta_I\}_{I=1}^n$ be the set of unknown parameters that we aim to estimate. We deduce a direct relation between the QFI matrix $\bm{\mathcal{F}}_{\mathcal{Q}}(\Theta)$, the covariance matrix of the channel $\bm V$, and the generator matrix $\bm{\mathcal{H}}$:
\begin{equation}\label{appeq:Vqfim}
    \sum_{I,J=1}^n\left(\bm{\mathcal{F}}_{\mathcal{Q}}(\Theta)\right)_{IJ}\vartheta_I\vartheta_J =4\Tr{\bm V\bm{\mathcal{H}}}.
\end{equation} 
We note that $\Tr{\bm V\bm{\mathcal{H}}}\sim\order{K^2\varepsilon^2}$, in accordance with the small parameter $\varepsilon$ introduced in Appendix~\ref{app:channel_derive}. To prove the result, we apply the following lemma.
\begin{lemma}\label{lemma:fidelity}
Let $\rho_{\bm V}=\Phi_{\bm V}(\rho)$ and assume an input pure state $\rho=\dyad{\psi}$ for the probe. The input-output fidelity of the quantum channel $\Phi_{\bm V}$ is approximately
    \begin{equation}\label{appeq:fidelity_V}
        F(\rho,\rho_{\bm V}) \approx 1-\Tr{\bm V\bm{\mathcal{H}}},
    \end{equation}
where $\bm{\mathcal{H}}$ is the Hamiltonian covariance matrix in the local basis,
    \begin{equation}\label{eq:hamiltonian_matrix}
        \bm{\mathcal{H}}_{ij}=\expval*{\hat{h}_i\hat{h}_j}-\expval*{\hat{h}_i}\expval*{\hat{h}_j},
    \end{equation}
with all expectation values evaluated with respect to the pure probe state $\rho=\dyad{\psi}$.
\end{lemma}

\begin{proof}
Recall the approximate form of the quantum channel $\Phi_{\bm V}$ in Eq.~(\ref{appeq:ch_expand}). Then, determine the fidelity between the input $\rho=\dyad{\psi}$ and the output $\rho_{\bm V}=\Phi_{\bm V}(\dyad{\psi})$:
    \begin{align}
        F(\rho,\rho_{\bm V})&= \expval{\Phi_{\bm V}\left(\dyad{\psi}\right)}{\psi} \\
        &\approx \ip{\psi}{\psi}^2 + \sum_{i,j} \bm V_{ij}\left(\expval{\hat{h}_i(\dyad{\psi})\hat{h}_j}{\psi} - \frac{1}{2}\expval{\acomm{\hat{h}_i\hat{h}_j}{\dyad{\psi}}}{\psi}\right) \\
        &= 1 + \sum_{i,j} \bm V_{ij}\left(\expval{\hat{h}_i}{\psi}\expval{\hat{h}_j}{\psi} - \frac{1}{2}\left(\expval{\hat{h}_i\hat{h}_j}{\psi}\ip{\psi}{\psi}+\ip{\psi}\expval{\hat{h}_i\hat{h}_j}{\psi}\right)\right) \\ 
        &= 1 + \sum_{i,j} \bm V_{ij}\left(\expval{\hat{h}_i}{\psi}\expval{\hat{h}_j}{\psi} - \expval{\hat{h}_i\hat{h}_j}{\psi}\right) \\ 
        &= 1 - \sum_{i,j} \bm V_{ij}\left(\expval{\hat{h}_i\hat{h}_j}- \expval{\hat{h}_i}\expval{\hat{h}_j}\right) \\
        &= 1- \Tr{\bm{\mathcal{H}}\bm V},
    \end{align}
where $\expval*{\hat{O}}=\expval*{\hat{O}}{\psi}$. 
\end{proof}

We now prove the main result. 
\begin{proof}
Recall the geometric relation between the QFI matrix and the fidelity via Bures distance~\cite{Liu2020QFIM} for any set of parameters $\{\varphi_I\}$, ${\sum_{I,J}\left(\bm{\mathcal{F}}_{\mathcal{Q}}\right)_{IJ}\delta\varphi_I\delta\varphi_J=8(1-\sqrt{F(\rho_\varphi,\rho_{\varphi+\delta\varphi})})}$,
where $F(\rho,\tau)={\rm Tr}[\sqrt{\sqrt{\rho}\tau\sqrt{\rho}}\,]^2$ is the fidelity. For one state pure, say $\tau=\dyad{\psi}$, the fidelity is simply the overlap, $F(\rho,\dyad{\psi})=\expval{\rho}{\psi}$. In our work, we estimate small (positive) fluctuations ($\vartheta>0$) from the identity map, so that
\begin{equation}\label{appeq:bures}
    \sum_{I,J=1}^n\left(\bm{\mathcal{F}}_{\mathcal{Q}}(\Theta)\right)_{IJ}\vartheta_I\vartheta_J=8(1-\sqrt{F(\rho_0,\rho_{\vartheta})}),
\end{equation}
where $\rho_0=\dyad{\psi}$ and $\rho_{\vartheta}=\Phi_{\bm V}(\rho_0)$. Using Lemma~\ref{lemma:fidelity}, we find that 
\begin{equation}
    F(\rho_0,\rho_\vartheta)\approx 1-\Tr{\bm V\bm{\mathcal{H}}} \implies 1-\sqrt{F(\rho_0,\rho_{\vartheta})} \approx \Tr{\bm V\bm{\mathcal{H}}}/2.
\end{equation}
Substituting the latter relation into Eq.~\eqref{appeq:bures}, we obtain the result.
\end{proof}

\section{No-Go Entanglement Advantage in Correlated Dephasing Estimation with Entangled States Generated by (Passive) Linear Optics}
\label{app:nogo_gaussian}
In this section, we derive a no-go result for entanglement advantage in estimating correlated (bosonic or fermionic) dephasing with entangled states generated by passive linear optics (e.g., entangled Gaussian bosonic states).

\begin{proposition}
    Entangled probes generated by propagating local, separable probes through a (passive) linear optical network perform no better than separable probes in estimating maximally correlated dephasing.
\end{proposition}
\begin{proof}
    Consider a passive linear-optical network, denoted by the unitary operator $\hat{U}_{\bm B}$, and a product of single-mode (Gaussian or non-Gaussian) resource states, $\ket{\psi_{\rm sep}}=\bigotimes_{i=1}^K\ket{\psi_i}$. Let $\ket{\psi_{\rm ent}}=\hat{U}_{\bm B}\ket{\psi_{\rm sep}}$ denote the multi-mode entangled state generated by linear optics. Now, define the average photon number operator per mode, $\hat{n}_{\rm avg}\coloneqq\sum_{i=1}^K\hat{n}_i/K$. Recall the QFI for dephasing estimation using entangled and separable strategies, respectively: 
     \begin{align}
        \mathcal{F}_{\mathcal{Q}}^{\rm ent}(g)&=4K^2{\rm Var}\left(\hat{n}_{\rm avg}\right),
        \\ 
        \mathcal{F}_{\mathcal{Q}}^{\,\rm sep}(g)&=4K\left(\sum_{i=1}^K {\rm Var}(\hat{n}_i)/K\right).
    \end{align}  
    Proving the proposition thus equates to demonstrating that ${\rm Var}(\hat{n}_{\rm avg})_{\psi_{\rm ent}}={\rm Var}(\hat{n}_{\rm avg})_{\psi_{\rm sep}}=\sum_{i=1}^K{\rm Var}(\hat{n}_i)_{\psi_i}/K^2$, where ${\rm Var}(\hat{n}_i)_{\psi_{i}}$ represent the occupation variances of the local states $\psi_{i}$. This follows straightforwardly because, for any state $\ket{\varphi}=\hat{U}_{\bm B}\ket{\phi}$,
    \begin{align}
        {\rm Var}(\hat{n}_{\rm avg})_{\varphi}&= \expval{\hat{n}_{\rm avg}^2}{\varphi}-\expval{\hat{n}_{\rm avg}}{\varphi}\hspace{-.5em}\expval{\hat{n}_{\rm avg}}{\varphi}\\
        &=\expval{\hat{U}_{\bm B}^\dagger\hat{n}_{\rm avg}^2\hat{U}_{\bm B}}{\phi}-\expval{\hat{U}_{\bm B}^\dagger\hat{n}_{\rm avg}\hat{U}_{\bm B}}{\varphi}\hspace{-.5em}\expval{\hat{U}_{\bm B}^\dagger\hat{n}_{\rm avg}\hat{U}_{\bm B}}{\varphi}\\
        &= \expval{\hat{n}_{\rm avg}^2}{\phi}-\expval{\hat{n}_{\rm avg}}{\phi}\hspace{-.5em}\expval{\hat{n}_{\rm avg}}{\phi}\\
        &={\rm Var}(\hat{n}_{\rm avg})_{\phi},
    \end{align}
    where the penultimate line derives from the fact that $\hat{U}_{\bm B}$ represents a passive transformation that conserves $\hat{n}_{\rm avg}$. Thus, ${\rm Var}(\hat{n}_{\rm avg})_{\psi_{\rm ent}}={\rm Var}(\hat{n}_{\rm avg})_{\psi_{\rm sep}}$, implying that the (separable) product state $\ket{\psi_{\rm sep}}=\bigotimes_i\ket{\psi_{i}}$, achieves equal performance as the entangled state $\ket{\psi_{\rm ent}}=\hat{U}_{\bm B}\ket{\psi_{\rm sep}}$.
\end{proof}

\begin{corollary*}
        Entangled Gaussian bosonic probes perform equally well as separable Gaussian bosonic probes in estimating maximally correlated bosonic dephasing.
\end{corollary*}
This follows from the fact that we may construct any (multi-mode) entangled Gaussian bosonic state by passing single-mode Gaussian (i.e., squeezed coherent) states through a passive linear optical network~\cite{Serafini17QCV}.

\section{Multiple Noise Parameters}
\label{app:multi-parameter}

In this section, we derive the QFI matrix [Eq.~\eqref{eq:multiparam_bd}] for multiple noise parameters; see the Claim in the main text. Consider an invertible transformation matrix $\bm W$, and define the new covariance
\begin{equation}
    \bm V^\prime\coloneqq \bm W\bm V\bm W^\top. 
\end{equation}
Parametrize $\bm V^\prime$ via 
\begin{equation}\label{eq:V_params}
    \bm V^\prime=
    \begin{pmatrix}
        \xi_1^2 & \xi_1\xi_2\mathscr{C}_{12} & \hdots & \xi_1\xi_K \mathscr{C}_{1K} \\
        \xi_1\xi_2\mathscr{C}_{12} & \xi_2^2 & \hdots & \xi_2\xi_K \mathscr{C}_{2K} \\
        \vdots & \vdots & \ddots & \vdots \\
        \xi_1\xi_K \mathscr{C}_{1K} & \xi_2\xi_K \mathscr{C}_{2K} & \hdots & \xi_K^2
    \end{pmatrix},
\end{equation}
where $\xi_I^2\coloneqq \left(\bm W \bm V\bm W^\top\right)_{II}$ denote (real) collective fluctuation parameters. The off-diagonal numbers $\mathscr{C}_{IJ}$ are residual correlation coefficients, which we take as nuisance parameters here. We thus focus on estimates of the collective fluctuations $\Xi=\{\xi_I\}_{I=1}^K$, for which we find the following [note that $\bm{W}^{-\top} = (\bm{W}^{\top})^{-1}$].
\begin{theorem}\label{thm:qfim}
    Consider parametrization $\bm V^\prime$ as in Eq.~\eqref{eq:V_params}. The QFI matrix for parameters $\Xi\coloneqq\{\xi_I\}_{I=1}^K$ is
        \begin{equation}\label{appeq:qfim}
        \left(\bm{\mathcal{F}}_{\mathcal{Q}}(\Xi)\right)_{IJ} =\begin{cases}
                4(\bm W^{-\top}\bm{\mathcal{H}}\bm W^{-1})_{II}, & I=J \\                8\mathscr{C}_{IJ}(\bm W^{-\top}\Re{\bm{\mathcal{H}}}\bm W^{-1})_{IJ}, & I\neq J,
            \end{cases}
        \end{equation}
    where $\Re{\bm{\mathcal{H}}}=(\bm{\mathcal{H}} + \bm{\mathcal{H}}^*)/2$. If $\bm W$ is orthogonal with $\bm W=(\vec{w}_1,\vec{w}_2,\dots)^\top$, then we rewrite the QFI matrix as
        \begin{equation}
            \left(\bm{\mathcal{F}}_{\mathcal{Q}}(\Xi)\right)_{IJ} =
            \begin{cases}
                4\left(\vec{w}_I^{\top}\bm{\mathcal{H}}\vec{w}_I\right), & I=J \\                8\mathscr{C}_{IJ}(\vec{w}_I^{\top}\Re{\bm{\mathcal{H}}}\vec{w}_J), & I\neq J
            \end{cases}.
        \end{equation}

\end{theorem}
\begin{proof}
    Recall the Bures distance~\eqref{appeq:bures}. Introduce the invertible matrix $\bm W$ into Eq.~\eqref{appeq:bures} to determine the following relations:
    \begin{align}
        \Tr{\bm V\bm{\mathcal{H}}} &= \Tr{(\bm W\bm V\bm W^\top)(\bm W^{-\top}\bm{\mathcal{H}}\bm W^{-1})}\\
        &= \Tr{\bm V^\prime(\bm W^{-\top}\bm{\mathcal{H}}\bm W^{-1})} \\
        &= \sum_{I,J}\bm V^\prime_{IJ}(\bm W^{-\top}\bm{\mathcal{H}}\bm W^{-1})_{IJ} \\
        &=\sum_{I}\xi_I^2(\bm W^{-\top}\bm{\mathcal{H}}\bm W^{-1})_{II} + \sum_{\substack{I,J \\ I\neq J}}\xi_I\xi_J\mathscr{C}_{IJ}(\bm W^{-\top}\bm{\mathcal{H}}\bm W^{-1})_{IJ}.
        \label{eq:Vh_prime}
    \end{align}
    Note that $\mathscr{C}_{JI}=\mathscr{C}_{IJ}$ and $\bm{\mathcal{H}}_{JI}=\bm{\mathcal{H}}_{IJ}^*$. Substituting Eq.~\eqref{eq:Vh_prime} into the right hand side of Eq.~\eqref{eq:qfim_Vh} and matching like terms on both sides of the equality, we deduce that 
    \begin{equation}
        \left(\bm{\mathcal{F}}_{\mathcal{Q}}(\Xi)\right)_{IJ} =\begin{cases}
                4(\bm W^{-\top}\bm{\mathcal{H}}\bm W^{-1})_{II}, & I=J \\                8\mathscr{C}_{IJ}(\bm W^{-\top}\Re{\bm{\mathcal{H}}}\bm W^{-1})_{IJ}, & I\neq J,
            \end{cases}
    \end{equation} 
    where $\Re{\bm{\mathcal{H}}}=(\bm{\mathcal{H}} + \bm{\mathcal{H}}^*)/2$. 
    
    Furthermore, suppose that $\bm W$ is orthogonal such that $\bm W^{\top}=\bm W^{-1}$, and write $\bm W=(\vec{w}_1, \vec{w}_2,\dots)^\top$. Then ${\bm W^{-\top}\bm{\mathcal{H}}\bm W^{-1}=\bm W\bm{\mathcal{H}}\bm W^{\top}}$ and
    \begin{align}
        \bm W\bm{\mathcal{H}}\bm W^{\top} &= 
        \begin{pmatrix}
            \vec{w}_1^\top \\
            \vec{w}_2^\top \\ 
            \vdots
        \end{pmatrix}
        \bm{\mathcal{H}}
        \begin{pmatrix}
            \vec{w}_1 & \vec{w}_2 & \hdots
        \end{pmatrix} \\
        & = 
        \begin{pmatrix}
            \vec{w}_1^\top\bm{\mathcal{H}}\vec{w}_1 & \vec{w}_1^\top\bm{\mathcal{H}}\vec{w}_2 & \hdots \\
            \vec{w}_2^\top\bm{\mathcal{H}}\vec{w}_1 & \vec{w}_2^\top\bm{\mathcal{H}}\vec{w}_2 & \hdots \\
            \vdots & \vdots & \ddots
        \end{pmatrix}.
    \end{align}
    Whence, $(\bm W^{-\top}\bm{\mathcal{H}}\bm W^{-1})_{IJ}=\vec{w}_I^\top\bm{\mathcal{H}}\vec{w}_J$,
    which concludes the proof.
\end{proof}

\section{Backgrounds Trigger Rayleigh's Curse}
\label{app:curse}
In this section, we extend our derivations to include background (e.g., thermal) fluctuations, $(\bm\Sigma)_{ij}\coloneqq\expval{\Delta\lambda_i\Delta\lambda_j}_{\rm bkg}$, that distinguish themselves from the desired parameters of $\bm V$. We assume the background fluctuations to arise from the same generators as $\bm V$, i.e. representing parallel decoherence~\cite{Sekatski2017Quantum_FullFast}. Practically, such backgrounds must be characterized and subtracted off in post-processing (a common practice when measuring feeble signals, such as searches for new physics~\cite{Bass2024NatRvw}). Regarding quantum limits, although backgrounds do not stymie an entanglement advantage for the correlated noise estimation problem, relatively large backgrounds do lead to the phenomenon of Rayleigh's curse~\cite{Tsang2016Superresolution, Gardner2024StochEst}, as we demonstrate below for the multi-parameter problem. 
\begin{proposition}\label{prop:rayleigh}
    Consider mixing background fluctuations, $\bm\Sigma$, into the signal, such that $\bm V_{ij}\rightarrow \bm\Sigma_{ij}+\bm V_{ij}$ and suppose both weak signal and weak background fluctuations. Given $\xi_I^2\coloneqq\bm (\bm W\bm V\bm W^{\top})_{II}$ and $\sigma_I^2\coloneqq(\bm W\bm{\Sigma}\bm W^{\top})_{II}$, the QFI for the collective fluctuations $\Xi=\{\xi_I\}_{I=1}^K$ is
    \begin{equation}\label{eq:rayleigh}
        (\bm{\mathcal{F}}_{Q}(\Xi;\sigma))_{II}=\left(\frac{4\xi_I^2}{\sigma_I^2+\xi_I^2}\right)(\bm W\bm{\mathcal{H}}\bm W^{\top})_{II}.
    \end{equation}
    The QFI suffers from Rayleigh's curse~\cite{Tsang2016Superresolution,Gardner2024StochEst}---i.e., the QFI vanishes (equivalently, the estimation error diverges) as $\xi_I\rightarrow 0$ for all probes and measurements given $\sigma_I>0$.
\end{proposition}
\begin{proof}[Proof sketch]
    Additive backgrounds, $\bm\Sigma$, change the noise channel to $\Phi_{\widetilde{\bm V}}$ with $\widetilde{\bm V}=\bm\Sigma+\bm V$. As before, we assume an approximate quantum channel $\Phi_{\widetilde{\bm V}}$, similar to Eq.~\eqref{appeq:ch_expand}, under the additional assumption of weak backgrounds. We then determine the QFI matrix, $\bm{\mathcal{F}}_{\mathcal{Q}}(\tilde{\Theta})$, of the new parameters $\tilde{\vartheta}_I^2\coloneqq(\bm W\widetilde{\bm V}\bm W^\top)_{II}=\sigma_I^2+\vartheta_I^2$, finding $\bm{\mathcal{F}}_{\mathcal{Q}}(\tilde{\Theta})=4(\bm W\bm{\mathcal{H}}\bm W^\top)_{II}$. Finally, we use $\pdv{\tilde{\vartheta}_I}{\vartheta_J}=\delta_{IJ}(\vartheta_J/\tilde{\vartheta}_J)$ and apply the error propagation rule to reckon ${\bm{\mathcal{F}}_{\mathcal{Q}}(\Theta)=\abs{\pdv{\tilde{\vartheta}_I}{\vartheta_I}}^2\bm{\mathcal{F}}_{\mathcal{Q}}(\tilde{\Theta})}$.
\end{proof}

Rayleigh's curse applies equally to the single-parameter estimation problems considered in the main text, i.e.~for $\bm V=g^2\bm v$. This fact follows the same line of reasoning as Proposition~\ref{prop:rayleigh}, but we spell out the details for clarity. 

Let $\widetilde{\bm V}=\bm V+\bm\Sigma$, where $\bm V=g^2\bm v$ and $\bm v=\Tr{\bm v}\vec{v}\vec{v}^\top$. Define the new parameter $\tilde{\vartheta}\coloneqq(\vec{v}^\top\widetilde{\bm V}\vec{v})^{1/2}=(g^2\Tr{\bm v}+\sigma^2)^{1/2}$, where $\sigma^2\coloneqq\vec{v}^\top\bm{\Sigma}\vec{v}$ denotes background fluctuations projected onto the $\vec{v}$ mode. To derive the QFI for $g$, we first find the QFI for $\tilde{\vartheta}$, which is simply $\mathcal{F}_{\mathcal{Q}}(\tilde{\vartheta})=4(\vec{v}^\top\bm{\mathcal{H}}\vec{v})$. Using the error propagation rule for the Fisher information and the fact that $\partial\tilde{\vartheta}/\partial g=g\Tr{\bm v}/\tilde{\vartheta}$, it follows that
\begin{align}
    \mathcal{F}_{\mathcal{Q}}(g)&= \bigg|\pdv{\tilde{\vartheta}}{g}\bigg|^2\mathcal{F}_{\mathcal{Q}}(\tilde{\vartheta}) \\
    &=\left(\frac{g^2\Tr{\bm v}^2}{g^2\Tr{\bm v}+\sigma^2}\right)4(\vec{v}^\top\bm{\mathcal{H}}\vec{v})\\ 
    &= \left(\frac{g^2\Tr{\bm v}}{g^2\Tr{\bm v}+\sigma^2}\right) 4\Tr{\bm v\bm{\mathcal{H}}}.
\end{align}
To arrive at the final equality, we used $\Tr{\bm v\bm{\mathcal{H}}}=\Tr{\bm v}\vec{v}^\top\bm{\mathcal{H}}\vec{v}$. Hence, the emergence of Rayleigh's curse. We note that, since the QFI still scales with the global quantity $\Tr{\bm v\bm{\mathcal{H}}}$, we nonetheless maintain an entanglement advantage over separable strategies, which differs from unitary estimation in the presence of (parallel) background decoherence~\cite{Sekatski2017Quantum_FullFast,Dobrza2017PRX_MarkovNoise,Zhou2018QECmetrology}.

\section{Multi-parameter Estimation of Many-Body Spin Dephasing Dynamics}
\label{app:spin-example}

In this section we give a concrete illustration of entanglement advantage and multi-parameter tradeoffs for collective-noise sensing with spin systems.  We first treat a pedagogical single-qubit example involving \emph{simultaneous} estimation of $X$ and $Z$ dephasing, which can be achieved with an entangled ancilla. We then explain how the same analysis extends to \emph{collective} dephasing in a $K$-qubit QSN, emphasizing which conceptual ingredients carry over without repeating technical steps.

This example highlights a general feature of multi-parameter noise estimation: even in the weak-noise regime where each parameter separately admits a constructive achievability scheme, the simultaneous estimation of the \emph{individual} single-parameter optima may require additional resources, such as ancillae.

\subsection{Single-qubit example}

Consider a single qubit exposed to independent stochastic noise processes generated by Pauli operators $X$ and $Z$, with unknown rates $\gamma_X$ and $\gamma_Z$, over a short interrogation time $t$.  In the weak-noise
(single-jump) regime,
\begin{equation}
    \rho(t)\;\approx\;\rho_0
    + t\Big(
        \gamma_Z\,\mathcal{D}[Z](\rho_0)
        + \gamma_X\,\mathcal{D}[X](\rho_0)
    \Big),
    \qquad
    \mathcal{D}[L](\rho) \coloneqq L\rho L^{\dag}- \frac{1}{2}\{ L^{\dag}L ,\rho\} ,
    \label{eq:singlequbit-XZ-shorttime}
\end{equation}
where we use $L^\dagger L= I $ for Pauli jumps.  For a pure probe $\rho_0=\dyad{\psi}$, this admits a decomposition:
\begin{equation}
    \rho(t)\;\approx\;
    \bigl(1-p_Z-p_X\bigr)\,\dyad{\psi}
    + p_Z\,\dyad{\psi_Z}
    + p_X\,\dyad{\psi_X}
    \;+\;O(t^2),
    \label{eq:singlequbit-XZ-branch}
\end{equation}
where
\begin{equation}
    \ket{\psi_Z}  \coloneqq  Z\ket{\psi},\qquad
    \ket{\psi_X}  \coloneqq  X\ket{\psi},
    \qquad
    p_Z=\gamma_Z t,\qquad p_X=\gamma_X t.
    \label{eq:singlequbit-XZ-branches-probs}
\end{equation}

\paragraph{Single-parameter optimality (review).}
If only $\gamma_Z$ is unknown (with $\gamma_X=0$), choosing $\ket{\psi}=\ket{+}$ yields $\ket{\psi_Z}=\ket{-}\perp\ket{+}$, where $\ket{\pm}=(\ket{\uparrow} \pm \ket{\downarrow})/\sqrt{2}$. Then the POVM $\{ M_+,\, I - M_+\}$ with $ M_+=\dyad{+}$ returns
\begin{equation}
    \Pr(\text{$Z$-jump})
    = 1-\Tr( M_+\rho(t))
    \;\approx\;
    \gamma_Z t.
\end{equation}
Hence, we infer $\gamma_Z$ at leading order from the Bernoulli statistics formed by measurement outcomes. This measurement strategy is optimal (i.e., achieves the highest possible Fisher information) over all probe states and measurement bases.  Similarly, if only $\gamma_X$ is nontrivial and unknown, choosing $\ket{\psi}=\ket{\uparrow}$ yields $X\ket{\psi}=\ket{\downarrow}\perp\ket{\uparrow}$ and the POVM $\{\dyad{\uparrow},\, I -\dyad{\uparrow}\}$ is optimal in the same sense.

\paragraph{Measurement compatibility for multi-parameter estimation.}
When both $\gamma_Z$ and $\gamma_X$ are unknown, any fixed two-outcome projective measurement of the above form mixes or misses jump events.  For example, with $\ket{\psi}=\ket{+}$ and $ M_+=\dyad{+}$, we find
\begin{align*}
    \Tr( M_+\rho(t))
    &\approx\;
    (1-p_Z-p_X)\underbrace{\Tr( M_+\dyad{\psi})}_{=1}
    + p_Z\underbrace{\Tr( M_+\dyad{\psi_Z})}_{=0}
    + p_X\underbrace{\Tr( M_+\dyad{\psi_X})}_{=1}
     \\
    &= 1-\gamma_Z t.
    \label{eq:singlequbit-mixing-example}
\end{align*}
This measurement optimally estimates $\gamma_Z$ but is
insensitive to $\gamma_X$.  Conversely, the choice $\ket{\psi}=\ket{\uparrow}$ estimates $\gamma_X$ but not $\gamma_Z$. Intuitively, attaining the \emph{individual} single-parameter optima for both rates requires a single POVM that distinguishes the three branches $\{\ket{\psi},\,Z\ket{\psi},\,X\ket{\psi}\}$ at leading order. This is impossible on a single qubit, since the Hilbert space is two-dimensional and cannot support three mutually orthogonal branch states.

\paragraph{Ancilla-assisted branch separation.}
The branch obstruction is lifted by enlarging the Hilbert space.  Consider a single ancilla qubit and prepare $\ket{+}_S\otimes\ket{\uparrow}_A$.  Apply an encoding unitary $\hat{U}$ (e.g., $\mathrm{CNOT}_{S\to A}$) satisfying
\begin{equation}
    U^\dagger(X\otimes  I )U = X\otimes X,
    \qquad
    U^\dagger(Z\otimes  I )U = Z\otimes  I .
    \label{eq:CNOT-conjugation}
\end{equation}
Conjugate the short-time noise channel by $\hat{U}$ to produce jump operators $L_Z=Z\otimes I $ and $L_X=X\otimes X$ acting on $S\otimes A$. To first order in $t$, the three relevant branches are
\begin{equation}
    \ket{+}\ket{\uparrow},\qquad
    (Z\otimes I )\ket{+}\ket{\uparrow} = \ket{-}\ket{\uparrow},\qquad
    (X\otimes X)\ket{+}\ket{\uparrow} = \ket{+}\ket{\downarrow},
\end{equation}
which are mutually orthogonal.  Hence, local POVMs
$\{\dyad{+}_S,\, I -\dyad{+}_S\}$ and
$\{\dyad{\uparrow}_A,\, I -\dyad{\uparrow}_A\}$ yield
\begin{equation}
    \Pr(\text{$Z$-jump}) = \gamma_Z t,\qquad
    \Pr(\text{$X$-jump}) = \gamma_X t,
    \label{eq:ancilla-probs-leading}
\end{equation}
Thus, we optimally infer the two rates simultaneously at leading order in $t$; moreover, rapidly repeating the experiment $\nu$ times, for a fixed total duration $T=\nu t$, remains optimal for all times~\cite{Sekatski2022:Thermometry}. This provides an explicit ancilla-assisted achievability scheme for simultaneous estimation, consistent with prior work highlighting the role of ancilla-assisted entanglement in multi-parameter noise estimation~\cite{Chen2022:PauliChEst}.

\paragraph{Time/ancilla tradeoff.}
This single-qubit example illustrates that multi-parameter incompatibility can be resolved in more than one way.  A \emph{serial} (time-sharing) strategy divides the total interrogation time $T$ between two single-parameter optimal measurements, each tailored to one noise rate. Alternatively, the \emph{ancilla-assisted} strategy enlarges the Hilbert space so that the relevant single-jump branches become orthogonal, enabling simultaneous estimation of both rates over the interrogation time.  In the weak-noise regime, this achieves twice the precision as the corresponding single-parameter strategy over the same $T$, at the cost of one ancillary qubit and an entangling operation. Thus, compatibility is restored by trading additional quantum resources for time efficiency. This resource tradeoff is consistent with prior results~\cite{Chen2022:PauliChEst} showing that ancilla-assisted entanglement can enable more favorable resource scaling for simultaneous estimation of multiple noise parameters than any serial strategy without ancillas.

\subsection{Extension to many-body, collective dephasing}

We now explain how the same mechanism extends to a QSN of $K$ qubits. Let $J_\alpha \coloneqq \sum_{k=1}^K \sigma_k^\alpha$ be the collective spin operators. Consider two-parameter collective dephasing,
\begin{equation}
    \rho(t)\;\approx\;\rho_0
    + t\Big(
        \gamma_{J_z}\,\mathcal{D}[J_z](\rho_0)
        + \gamma_{J_x}\,\mathcal{D}[J_x](\rho_0)
    \Big),
    \label{eq:collective-JxJz-shorttime}
\end{equation}
which is the many-body analogue of Eq.~\eqref{eq:singlequbit-XZ-shorttime}.  For a pure probe $\rho_0=\dyad{\psi}$, the single-jump branch picture becomes
\begin{equation}
    \rho(t)\;\approx\;
    \bigl(1-p_z-p_x\bigr)\,\dyad{\psi}
    + p_z\,\frac{\dyad{e_z}}{\ip{e_z}{e_z}}
    + p_x\,\frac{\dyad{e_x}}{\ip{e_x}{e_x}}
    \;+\;O(t^2),
    \label{eq:collective-branch-mixture}
\end{equation}
with centered jump branches
\begin{equation}
    \ket{e_z} \coloneqq \Delta J_z\ket{\psi},\qquad
    \ket{e_x} \coloneqq \Delta J_x\ket{\psi},
\end{equation}
where $\Delta J_i \coloneqq J_i - \expval{J_i}$, and probabilities
\begin{equation}
    p_z=\gamma_{J_z}t \expval{\Delta J_z^2}{\psi}
    \qquad     p_x=\gamma_{J_x}t \expval{\Delta J_x^2}{\psi}.
    \label{eq:collective-pz-px}
\end{equation}
Thus, at leading order, the two rates again enter only through the two branch weights $p_z$ and $p_x$.

\paragraph{Single-parameter optimality and entanglement scaling.}
Before addressing the multi-parameter setting, we address entanglement advantage in the single-parameter scenario. Consider estimating $\gamma_{J_z}$ alone, with $\gamma_{J_x}=0$.  In the weak-noise regime, the collective fluctuation $\Delta J_z^2$ determine the jump probability, viz., $p_z=\gamma_{J_z}t\expval*{\Delta J_z^2}{\psi}$.  For product probes, $\expval*{\Delta J_z^2}=\order{K}$, revealing shot-noise scaling. By contrast, for a suitably chosen many-body entangled probe (e.g., GHZ-type state), we have $\expval*{\Delta J_z^2}=\order{K^2}$, revealing Heisenberg scaling.  In this single-parameter setting, the POVM $\{P_\psi \coloneqq \ketbra{\psi},I-P_\psi\}$ provides a constructive achievability scheme in the weak-noise regime.  A similar discussion applies to estimating $\gamma_{J_x}$ alone.  Thus, entanglement between the probe spins is advantageous and essential for achieving Heisenberg scaling in collective-noise estimation.

\paragraph{Measurement compatibility for multi-parameter estimation.}
As in the single-qubit case, a measurement that is optimal for estimating $\gamma_{J_z}$ when $\gamma_{J_x}=0$ need not remain optimal when both rates are present.  For example, the POVM $\{ M_\psi, I - M_\psi\}$ with $ M_\psi=\dyad{\psi}$ detects whether \emph{any} jump occurred but does not distinguish between $J_z$ and $J_x$ jumps.  Indeed,
\begin{equation}
    \Tr{(I-M_\psi)\rho(t)}\approx p_z+p_x,
\end{equation}
which depends only on the sum of rates weighted by the corresponding variances in Eq.~\eqref{eq:collective-pz-px}.  Hence, separating $\gamma_{J_z}$ and $\gamma_{J_x}$ generally requires either a larger POVM that resolves distinct jump subspaces, multiple experimental settings, or additional degrees of freedom (entangled ancillae) that render the branches distinguishable in an enlarged Hilbert space.

\paragraph{Collective ancilla-assisted branch separation.}
The obstruction above has the same origin as in the single-qubit case: in the weak-noise regime, simultaneous estimation of $\gamma_{J_z}$ and $\gamma_{J_x}$ requires distinguishing the three branches $\{\ket{\psi},\,J_z\ket{\psi},\,J_x\ket{\psi}\}$.  For sensor-only probes, uncertainty relations of fluctuations $\expval*{\Delta J_x^2}{\psi}$ and $\expval*{\Delta J_z^2}{\psi}$ of the non-commuting operators indicate that optimality between the two rates cannot be achieved simultaneously. As such, measurements acting solely on the QSN cannot, in general, separate these branches without sacrificing sensitivity. 

Enlarging the Hilbert space with ancillae lifts this obstruction by providing extra degrees of freedom that render the collective jump branches orthogonal while preserving large collective fluctuations for \emph{each} noise process, thereby enabling simultaneous estimation at the Heisenberg limit. For this purpose, we introduce an ancilla register of dimension $K+1$ (e.g., corresponding to $\lceil\log_2(K+1)\rceil$ ancilla qubits) and consider preparing the \emph{Dicke-Choi probe}~\footnote{We introduce this nomenclature because $\ket{\Phi_{\rm sym}}$ is the Choi representation of the projection onto the maximally symmetric subspace spanned by Dicke states, $P_{\rm sym}\propto \sum_{n=0}^K \dyad*{D_n^{(K)}}$.}
\begin{equation}\label{eq:dickechoi}
    \ket{\Phi_{\mathrm{sym}}}
    =\frac{1}{\sqrt{K+1}}
    \sum_{n=0}^{K}\ket*{D^{(K)}_{n}}_S\ket{n}_A,
\end{equation}
where $\ket*{D^{(K)}_{n}}$ denotes the $K$-qubit Dicke state with $n$ excitations~\cite{Toth2012:dicke, Pezze2018RMP_atomic} and $\ket{n}$ represents the ancilla register~\footnote{The Dicke-Choi probe~\eqref{eq:dickechoi} has somewhat similar properties to the (ancilla-free) quantum compass states considered in Ref.~\cite{Vasilyev2024:QuCompass} for phase estimation.}. We state several useful facts about Dicke states:
\begin{itemize}
    \item The Dicke state $\ket*{D^{(K)}_{n}}$ is the equal superposition of all $K$-qubit computational basis states with exactly $n$ excitations (i.e., Hamming weight $n$).
    \item Dicke states coincide with the collective spin eigenstates $\ket*{J=K/2, m=K/2-n}$, such that $J_z\ket*{D_n^{(K)}}=(K-2n)\ket*{D_n^{(K)}}$ in our convention $J_z=\sum_{k=1}^K \sigma_k^{z}$.
    \item The Dicke states $\{\ket*{D^{(K)}_{n}}\}_{n=0}^K$ form an orthonormal basis of the fully symmetric subspace of dimension $K+1$. Consequently, the Dicke-Choi probe $\ket{\Phi_{\rm sym}}$ is a purification of the symmetric subspace projector, $\Tr_A{\dyad{\Phi_{\rm sym}}}=P_{\rm sym}/(K+1)$.
    \item High-weight Dicke states can, in principle, be prepared efficiently using collective control protocols~\cite{Yu2025:DickePrep}.
\end{itemize}

For the Dicke-Choi probe~\eqref{eq:dickechoi}, the reduced state on the QSN is maximally mixed over the Dicke ladder (i.e., equivalent to the normalized projector on the symmetric subspace) implying large and isotropic fluctuations,
\begin{equation}
    \expval{\Delta J_x^2}{\Phi_{\mathrm{sym}}}
    =\expval{\Delta J_z^2}{\Phi_{\mathrm{sym}}}
    =\frac{K(K+2)}{3}.
\end{equation}
Furthermore, the collective jump branches $\ket{e_\alpha}=(\Delta J_\alpha\otimes I)\ket{\Phi_{\mathrm{sym}}}$ are orthogonal $\ip{e_x}{e_z}\propto \Tr{P_{\rm sym}\Delta J_x\Delta J_z}=0$. Consequently, the Fisher information matrix of this measurement scheme is diagonal with entries $\mathcal{F}(\gamma_{J_\alpha})
\approx t\,\expval{\Delta J_\alpha^2}{\Phi_{\mathrm{sym}}}/\gamma_{J_\alpha}
= \order{K^2}$, which follows directly from the branch probabilities, $p_\alpha \approx \gamma_\alpha t \expval{\Delta J_\alpha^2}{\Phi_{\mathrm{sym}}}$. 

This measurement strategy saturates the quantum Cram\'er-Rao bound to leading order in $t$ and also connects naturally to our general QFI framework.  For independent collective dephasing rates, the noise covariance is
\begin{equation}
    \bm V_{\alpha\beta} = \gamma_{J_\alpha}t\,\delta_{\alpha\beta},
\end{equation}
and the Hamiltonian (generator) matrix is
\begin{equation}
    \bm{\mathcal{H}}_{\alpha\beta}
    = \expval{\Delta J_\alpha\,\Delta J_\beta}{\Phi_{\mathrm{sym}}}.
\end{equation}
Using our main result $\sum_{\alpha\beta}(\bm{\mathcal{F}}_Q)_{\alpha\beta}\vartheta_\alpha\vartheta_\beta\approx 4\Tr{\bm V\bm{\mathcal{H}}}$ [Eq.~\eqref{appeq:Vqfim}] for parametrization $\vartheta_\alpha=\sqrt{\gamma_{J_\alpha}}$, we deduce that the QFI matrix for $\{\vartheta_\alpha\}$ is diagonal with entries $\mathcal{F}_Q(\vartheta_\alpha=\sqrt{\gamma_{J_\alpha}})\approx 4t \bm{\mathcal{H}}_{\alpha\alpha}$. Whence by the parameter change rule, 
\begin{equation}
 \mathcal{F}_Q(\gamma_{J_\alpha})
    \approx \frac{t\,\bm{\mathcal{H}}_{\alpha\alpha}}{\gamma_{J_\alpha}}
    = \frac{t\,\expval{\Delta J_\alpha^2}{\Phi_{\mathrm{sym}}}}{\gamma_{J_\alpha}},
\end{equation}
in agreement with the classical Fisher-information above. Repeating this short-time experiment $\nu$ times, for a total duration $T=\nu t$, yields the accumulated QFI $\mathcal{F}_Q = \order{T K^2}$. Thus, an ancilla register of dimension $K+1$ (e.g., using $\lceil\log_2(K+1)\rceil$ ancilla qubits) suffices to render the collective jump branches distinguishable, achieving the quantum Cram\'er-Rao bound and yielding Heisenberg scaling for both rates in the weak-noise regime.

\end{document}